\documentclass[journal,twoside,web]{ieeecolor}
\usepackage{generic}
\usepackage{cite}
\usepackage{amsmath,amssymb,amsfonts}
\usepackage{algorithmic}
\usepackage{graphicx}
\usepackage{algorithm,algorithmic}
\usepackage{hyperref}
\hypersetup{hidelinks=true}
\usepackage{textcomp}
\usepackage{siunitx}
\def\BibTeX{{\rm B\kern-.05em{\sc i\kern-.025em b}\kern-.08em
    T\kern-.1667em\lower.7ex\hbox{E}\kern-.125emX}}
\usepackage{graphicx}

\usepackage{amsmath} % assumes amsmath package installed
\usepackage{amsthm}

\usepackage{amsfonts}
\usepackage{amssymb}  % assumes amsmath package installed
\usepackage{mathtools} %% for \coloneqq
\usepackage{mathrsfs}
\usepackage{balance}
\DeclareMathAlphabet{\mathpzc}{OT1}{pzc}{m}{it}

\usepackage{url}
\usepackage{hyperref}

\usepackage[utf8]{inputenc}
\usepackage[T1]{fontenc}%\usepackage{clomment}
\usepackage{cite}

\usepackage{multirow}
\usepackage{multicol}
\let\labelindent\relax
\usepackage{enumitem}
\newcommand{\R}{\mathbb{R}}

\theoremstyle{definition}
\newtheorem{proposition}{Proposition}
\newtheorem{definition}{Definition}
\newtheorem{corollary}{Corollary}
\newtheorem{lemma}{Lemma}
\newtheorem{theorem}{Theorem} 
\newtheorem{remark}{Remark}

\newtheorem{condition}{Condition} 

\usepackage{siunitx}

\usepackage{subcaption}

\usepackage{indentfirst}
\usepackage{bm}
\newcommand{\cbf}{h}
\newcommand{\traj}{\bm{x}}
\newcommand{\ctrl}{\bm{u}}
\newcommand{\trajeta}{\bm{\eta}}
\newcommand{\trajphi}{\bm{\phi}}
\newcommand{\trajmu}{\bm{\mu}}
\let\phi\relax
\newcommand{\phi}{\varphi}
\newcommand{\phid}{\phi_{d}}

\newcommand{\trajphid}{\bm{\phid}}
\newcommand{\phidiff}{\Delta\phi}
\newcommand{\trajphidiff}{\bm{\phidiff}}
\newcommand{\XiToPsi}{T}
\newcommand{\barrierconstraint}{barrier constraint}
\newcommand{\safeset}{\mathcal{C}}

\newcommand{\dtheta}{\dot{\theta}}
\newcommand{\ddtheta}{\ddot{\theta}}

\newcommand{\pspi}{\kappa_{\text{ps}}}

\newcommand{\norm}[1]{\left\lVert #1\right\rVert}

\newcommand{\qio}{p_{\text{io}}}
\newcommand{\pcbf}{p_{\text{cbf}}}
\newcommand{\qcbf}{q_{\text{cbf}}}

\begin{document}
\title{When are safety filters safe? On minimum phase conditions of control barrier functions}
\author{Jason J. Choi, \IEEEmembership{Member, IEEE}, Claire J. Tomlin, \IEEEmembership{Fellow, IEEE},  Shankar Sastry, \IEEEmembership{Fellow, IEEE}, \\ and Koushil Sreenath, \IEEEmembership{Member, IEEE}
\thanks{% \jcnote{TODO: Date of submission.} 
This work is supported in part by NSF CMMI-1944722.}
\thanks{Jason J. Choi, Claire J. Tomlin, Shankar Sastry, and Koushil Sreenath are with University of California, Berkeley. Contact info: \tt jason.choi@berkeley.edu}}

\maketitle

\begin{abstract} 
In emerging control applications involving multiple and complex tasks, safety filters are gaining prominence as a modular approach to enforcing safety constraints. Among various methods, control barrier functions (CBFs) are widely used for designing safety filters due to their simplicity---imposing a single linear constraint on the control input at each state. In this work, we focus on the internal dynamics of systems governed by CBF-constrained control laws. Our key observation is that, although CBFs guarantee safety by enforcing state constraints, they can inadvertently be ``unsafe'' by causing the internal state to diverge. We investigate the conditions under which the full system state, including the internal state, can remain bounded under a CBF-based safety filter. Drawing inspiration from the input-output linearization literature, where boundedness is ensured by minimum phase conditions, we propose a new set of CBF minimum phase conditions tailored to the structure imposed by the CBF constraint. A critical distinction from the original minimum phase conditions is that the internal dynamics in our setting is driven by a nonnegative virtual control input, which reflects the enforcement of the safety constraint. We include a range of numerical examples---including single-input, multi-input, linear, and nonlinear systems---validating both our analysis and the necessity of the proposed CBF minimum phase conditions.
\end{abstract}


\vspace{-0.5em}
\section{Introduction}
\label{sec:introduction}

\subsection{Motivation \& Related Work}

\IEEEPARstart{S}{afety} filters are emerging as an effective paradigm for designing safe control systems  \cite{Wabersich2023, hsu2023safety}. Unlike conventional methods that integrate safety and performance into a single controller, the safety filter approach introduces a separate module as the final layer before the plant, which checks whether the commanded input may violate safety constraints and, if necessary, minimally overrides it. By decoupling safety from performance, this modular design offers simplicity and flexibility, making it well-suited for emerging applications involving complex and multi-task control \cite{he2024agile, lavanakul2024safety}.

% This approach decouples the design for performance from design for safety. % with the latter only intervening as needed to ensure safety. 
% The designer does not need to reconcile safety specifications while developing task-specific controllers, allowing the safety problem to be treated independently from the task problem. Such modular design process offers both simplicity and flexibility, making the safety filter approach well-suited for emerging applications that involve complex, multi-task control systems \cite{he2024agile, lavanakul2024safety}. 

Control barrier functions (CBFs) \cite{Prajna2006, ogren2006autonomous, Ames2016} represent a popular method for designing safety filters, alongside other methods including Hamilton-Jacobi reachability \cite{bansal2017hamilton} and predictive control \cite{wabersich2021predictive}. CBFs have gained popularity because they streamline the safety filter design process by imposing a single linear constraint on the control input at each state---referred to in this paper as the \textit{barrier constraint}. Several studies have identified potential drawbacks or issues with CBF-based safety filter designs \cite{reis2020control, grover2021deadlock, cortez2022compatibility, borquez2024safety, mestres2025control}, such as undesirable equilibria or deadlocks, which can compromise the performance of the original controller. 

In this work, we examine the often-overlooked \textit{internal dynamics} of CBF-based safety-filtered systems. The internal state arises as leftover state variables when the CBF's relative degree is lower than the system order. While CBFs ensure compliance with pre-specified safety constraints, our key observation is that they can introduce a different kind of failure by causing the internal state to diverge or blow up, which would not occur in the absence of the filter. For a motivating example illustrating the use of a CBF safety filter on a linear dynamical system leading to an unbounded state under some conditions, see Section \ref{subsec:linear-system-result}. 

We investigate conditions under which the full system state, including its internal state, remains bounded under a CBF-constrained control law. Recent work \cite{mestres2025control} analyzed boundedness in a special case where the CBF has relative degree one and a bounded zero-level set boundary. Here, we consider a more general setting with potentially higher relative degree and unbounded zero-level set boundaries.

 % Our analysis reveals that naive CBF-based safety filtering does not guarantee the boundedness of the internal state. 
We draw inspiration from the input-output (IO) linearization problem \cite[Ch. 9]{sastry2013nonlinear}, in which the conditions resulting in the boundedness of the system state are carefully studied. In IO linearization, first, the \textit{zero dynamics} is defined as the internal dynamics under which the output is perfectly regulated as zero. Next, the \textit{minimum phase condition} is defined as the condition under which the zero dynamics has a stable equilibrium. When this minimum phase condition is satisfied, the system state stays bounded under the IO linearization control law.

There are two key aspects that distinguish our problem from the IO linearization problem:
\begin{enumerate}[label=\alph*), leftmargin=1.25em, labelindent=\parindent, listparindent=\parindent, labelwidth=0pt]
    \item In IO linearization, the control law is designed to drive the output to converge to zero, whereas the CBF filter merely regulates the output (the CBF value) to remain nonnegative.
    \item As a result, when analyzing the internal dynamics, we must account for how the output influences the internal dynamics, particularly given the nonnegativity constraint on the output. Such a nonnegativity constraint is not present in the IO linearization control law, where the output variable is allowed to vary between positive and negative values.
\end{enumerate}
These aspects relate to the \textit{positive stabilization} problem discussed in the literature \cite{heemels1998positive}, wherein a system is positively stabilizable if it can be stabilized to an equilibrium using a nonnegative control signal or feedback control law. 

\vspace{-0.5em}
\subsection{Contributions}

These key differences motivate us to propose a new set of minimum phase conditions and zero dynamics tailored to the CBF-constrained system, distinct from those in the IO linearization problem. In our new definition of CBF zero dynamics (Definition \ref{def:zero-dynamics-cbf}), a virtual control input appears as a scalar input in the dynamics. Our analysis shows that, for a CBF whose zero dynamics is positively stabilizable, the boundedness of the CBF-constrained dynamics can be guaranteed under additional conditions. 

Accordingly, we define the new CBF minimum phase condition as the condition in which the given CBF results in positively stabilizable zero dynamics. This new definition generalizes the original definition of the minimum phase condition, as open-loop stability implies positive stabilizability. We specifically define three different types of CBF minimum phase conditions (Definitions \ref{cond:minimum-phase-simple}, \ref{cond:minimum-phase-feedback}, and \ref{cond:minimum-phase-multi-input}) for both single-input and multi-input systems.
% \begin{enumerate}[leftmargin=1.25em, labelindent=\parindent, listparindent=\parindent, labelwidth=0pt]
%     \item  \textit{Fixed-virtual-control-input} minimum phase condition (Condition \ref{cond:minimum-phase-simple}),
%     \item \textit{Positive-feedback-stabilized} minimum phase condition (Condition \ref{cond:minimum-phase-feedback}),
%     \item \textit{Open-loop-positive-stabilized} minimum phase condition (Condition \ref{cond:minimum-phase-open-loop}).
% \end{enumerate}
If the CBF satisfies one of these conditions, the CBF-constrained dynamics can be guaranteed to remain bounded, provided certain additional requirements are met (for instance, that the CBF minimum decay rate is sufficiently large in Theorem \ref{thm:min-phase-feedback}).
This result is established in Theorems \ref{thm:min-phase-simple-global} and \ref{thm:min-phase-feedback}, which together represent the main contribution of our work. 

The main takeaway from our analysis is that naive CBF safety filters can be problematic and potentially dangerous. \textit{CBF safety filters must be designed with care:} we should verify that the chosen CBF satisfies the CBF minimum phase conditions, and the stability of the internal dynamics must be addressed concurrently with the barrier constraint when employing the safety filter.

\vspace{-1em}
\subsection{Organization \& Notations}
The rest of the article is organized as follows. In Section \ref{sec:background}, we review the IO linearization problem and CBFs with general relative degree. Section \ref{sec:minimum-phase-cbf} presents our main results on the CBF minimum phase conditions for single-input systems. Section \ref{sec:multi-input} extends our analysis to multi-input systems. Section \ref{sec:simulations} presents a simulation study for the single-input case. Section \ref{sec:multi-input-examples} extends the simulation study to the multi-input case. We conclude with closing remarks in Section \ref{sec:conclusion}.

\textbf{Notations.} We use bold symbols when we consider them as signals in time. $|\cdot|$ indicates vector two-norm and $\|\cdot\|$ indicates matrix two-norm.

\vspace{-0.5em}

\section{Background}
\label{sec:background}

% \jcnote{We will write everything in single-input single-output case and will consider multi-input in the later section. We will keep our exposition for single-output case since we are considering a single CBF throughout the paper.}

Consider a control-affine system
\begin{equation}
    \dot{\traj}(t) = f(\traj(t)) + g(\traj(t)) \ctrl(t),\quad \traj(0) = x_0,
    \label{eq:system}
\end{equation}
where $\traj(t)\!\in\!\R^n$, $\ctrl(t)\!\in\!\R^m$, $f:\!\R^n\!\rightarrow\!\R^n$, $g:\!\R^n\!\rightarrow\!\R^{n \times m}$. We assume that $f, g$ are Lipschitz continuous with respect to $x$ and $u$, which is required for the forward completeness of the trajectory. We consider the single-input case ($m\!=\!1$) in Section \ref{sec:minimum-phase-cbf} and consider the multi-input case ($m\!\ge\!2$) in Section \ref{sec:multi-input}.

Next, we consider a scalar function $h:\R^n \rightarrow \R$ of the state $x$, which is considered an output function in the context of IO linearization with the goal of driving $h(x)$ to zero, and it is considered the CBF in the context of CBF-constrained dynamics with the goal of maintaining $h(x)$ nonnegative for all time. We will keep our exposition for this single-\textit{output} case throughout the paper since we focus on a single CBF constraint in this work.

Given the function $h$, we define the notion of relative degree.

\begin{definition}[Relative degree]
\label{def:relative-degree}
A continuously differentiable function $h:\R^n \rightarrow \R$ is said to \textit{have a relative degree} $r$ \textit{at} $x_0 \in \R^n$ if for all $i = 1, \cdots, m$, and $k = 0, \cdots, r-2$,
\begin{equation}
    L_{g_i} L_f^{k} h(x) = 0,
\end{equation}
for all $x$ in an open set around $x_0$, and the row vector
\begin{equation}
    L_{g} L_f^{r-1} h(x) = \left[ L_{g_1} L_f^{r-1} h(x) \quad \cdots \quad L_{g_m} L_f^{r-1} h(x) \right]
\end{equation}
has a non-zero element at $x_0$.    
\end{definition}
\vspace{-0.25em}
In this work, we consider the case where $r < n$ in which the internal dynamics will appear as we will see next.

\vspace{-0.5em}
\subsection{Input-Output (IO) Linearization}
\label{subsec:io-linearization}

Now, we briefly introduce the IO linearization problem, where one is interested in regulating the output $y:=h(x)$ to zero.
% We focus on the single input single output (SISO) problem since when we move on to CBF, we only consider a single CBF constraint.

If $h$ has relative degree $r$, we have 
\begin{equation}
    y^{(k)}=L^{k}_f h(x) \;\;\text{for all} \;\;k < r,
    \label{eq:io-linearization-output-derivative}
\end{equation}
where $u$ does not appear, and the first time the control term appears in the high-order time derivative of $y$ is 
\begin{equation}
    y^{(r)} = L_f^{r} h(x) + L_g L_f^{r -1} h(x) u.
    \label{eq:io-linearization}
\end{equation}
We define \vspace{-0.25em}
\begin{equation}
\label{eq:output-vector}
    \xi(x) := \left[\cbf(x), L_f \cbf(x), \cdots, L_f^{r-1} \cbf(x) \right]^\top \in \R^r, \vspace{-0.25em}
\end{equation}
as the \textit{output derivative vector} of $\cbf$, and $\nu : = y^{(r)}$ as the \textit{virtual control input} for IO linearization.
% \begin{equation}
%     \nu : = y^{(\gamma)} = L_f^{\gamma} h(x) + L_g L_f^{\gamma -1} h(x) u,
%     \label{eq:io-linearization}
% \end{equation}
% IO linearization rule:
% \begin{equation}
%     u = \frac{1}{L_g L_f^{\gamma -1} h(x)}\left(-L_f^{\gamma} h(x) + \nu \right),
%     
% \end{equation}
%where $\nu$ is the virtual control input for IO linearization. 
An appropriate $u$ that sets $\nu$ to a desired value can always be selected, since $L_g L_f^{r -1} h(x)$ has a non-zero element according to Definition \ref{def:relative-degree}. In particular, in the single-input case, the IO linearization rule can be determined as \vspace{-0.25em}
\begin{equation}
u = (L_g L_f^{r - 1} h(x))^{-1}(-L_f^r h (x) + \nu).
\label{eq:io-linearization-siso} \vspace{-0.25em}
\end{equation}

% Combining \eqref{eq:io-linearization-output-derivative} and the definition of $\nu$, we get the dynamics of the output derivative vector $\xi$ as
% \begin{equation}
%     \dot{\xi} = A \xi + B \nu,
% \label{eq:io-output-dynamics}    
% \end{equation}
% with
% {\small
% \begin{equation*}
%     A := \begin{bmatrix} 0 & 1 & 0 & \cdots & 0 \\
%     0 & 0 & 1 & \cdots & 0 \\
%     \; & \; & \ddots & \ddots & \; \\
%     0 & \; & \cdots & 0 & 1 \\    
%     0 & \; & \cdots & 0 & 0  \end{bmatrix},\;\; B = \begin{bmatrix} 0 \\ \vdots \\ 0 \\ 1\end{bmatrix}.
% \end{equation*}}

To regulate the output to zero, we can design $\nu$ as the linear feedback of the output derivative vector: \vspace{-0.25em}
\begin{equation}
    \nu = - k^\top \xi(x),
    \label{eq:io-virtual-input-law} \vspace{-0.25em}
\end{equation}
where $k \in \R^{r}$ is the feedback gain vector. Under \eqref{eq:io-virtual-input-law}, the closed-loop dynamics of the output becomes \vspace{-0.25em}
\begin{equation}
    \dot{\xi} = A_k \xi,
    \label{eq:io-output-dynamic-closed} \vspace{-0.25em}
\end{equation} 
where
\begin{equation}
    A_k := \resizebox{.4\hsize}{!}{$\displaystyle \begin{bmatrix} 0 & 1 & 0 & \cdots & 0 \\
    0 & 0 & 1 & \cdots & 0 \\
    \; & \; & \ddots & \; & \; \\
    0 & 0 & \cdots & 0 & 1 \\    
    -k_1 & -k_2 & \cdots & -k_{r-1} & -k_r \end{bmatrix}$}. \vspace{-0.25em}
\label{eq:A_xi}
\end{equation}

\noindent The feedback gain $k$ is chosen such that $A_k$ and the polynomial $s^r \!+\!k_r s^{r-1}\!+\!\cdots \!+\!k_1$ are Hurwitz, resulting in the output dynamics to be stabilized to the equilibrium $\xi_e\!=\!0$ where $y$ becomes zero.

Since the output dynamics \eqref{eq:io-output-dynamic-closed} is defined in an $r$-dimensional space, there are $n\!-\!r$ remaining dimensions of the state that are not described by the output dynamics. Note that if $r=n$, we do not have such remaining dimensions and can achieve full-state linearization \cite[Sec. 9.2.1]{sastry2013nonlinear}. To define the coordinates that describe these remaining dynamics, we consider a coordinate system transformation that is a diffeomorphism, defined by $n-r$ functions $\eta_1(x), \cdots, \eta_{n-r}(x)$ satisfying the following condition:
\begin{condition}[Valid diffeomorphic change of coordinates]\label{cond:nonsingular} The Jacobian of $[\xi^\top, \eta^\top]^\top$, \vspace{-0.5em}
$$
\begin{bmatrix} d \xi \\ d \eta \end{bmatrix} \in \mathbb{R}^{n \times n}, \vspace{-0.25em}
$$
is nonsingular at $x_0$. Under this condition, a local coordinate transformation $\Psi_\xi: x \rightarrow [\xi^\top, \eta^\top ]^\top$ is a local diffeomorphism at a small neighborhood of $x_0$ \cite[Theorem 3.35]{sastry2013nonlinear}.
\end{condition}

If Condition \ref{cond:nonsingular} is satisfied, we consider $\eta \in \R^{n-r}$ as the \textit{internal state} at $x_0$. We discuss the minimum phase conditions for IO linearization for the single-input case in Section \ref{subsec:single-input-io}.

\vspace{-0.5em}

\subsection{Control Barrier Functions with General Relative Degree}
\label{subsec:cbf}

Next, we introduce the concept of Control Barrier Functions with general relative degree, that can be used to design a safety filter for the system \eqref{eq:system}. Consider a safety constraint that requires the trajectory to stay inside a \textit{constraint set} $\safeset$ indefinitely, given as 
\begin{equation}
\traj(t) \in \safeset:= \{x\in\R^n \;|\; \cbf(x) \ge 0\}, \;\;\text{for all}\;\; t \ge 0,
\label{eq:safety-spec}
\end{equation}
where $\safeset$ is given as the zero-superlevel set of $h$. If $h$ satisfies the following definition, we can design a CBF-safety filter that ensures \eqref{eq:safety-spec}.

\begin{definition}[Exponential Control Barrier Function \& barrier constraint \cite{nguyen2016exponential}] \label{def:ho-cbf} A function $\cbf: \R^n \rightarrow \R$ is called an \textit{(exponential)} \textit{CBF with relative degree $r$} if $h$ has a relative degree $r$ for all $x \in \safeset= \{x \;|\;\cbf(x) \ge 0\}$, $\frac{\partial h}{\partial x}\neq 0$ for all $x \in \{x\; | \; \cbf(x) = 0\}$, and
\begin{equation}
    \sup_{u \in \R^m} \big[ \underbrace{L_f^r \cbf(x) +  L_g L_f^{r - 1} \cbf(x) u}_{\cbf^{(r)}(x, u)} + k^\top \xi(x) \big] \ge 0
\end{equation}
is satisfied for all $x \in \safeset$, where $\xi$ is the \textit{CBF output derivative vector} defined in \eqref{eq:output-vector} and $k$ is a constant $r$-dimensional vector $k = [k_1, \cdots, k_r]^\top$ such that the roots of the polynomial $s^r + k_r s^{r-1} + \cdots + k_1$ are negative reals.
% , where 
% \begin{equation}
%     \xi(x) := \left[\cbf(x), L_f \cbf(x), \cdots, L_f^{r-1} \cbf(x) \right]^\top \in \R^r.
% \end{equation}
We also define
\begin{equation}
\mu(x, u) := L_f^r \cbf(x) +  L_g L_f^{r - 1} \cbf(x) u + k^\top \xi(x) \ge 0,
\label{eq:barrier-constraint}
\end{equation}
as the (exponential) \textit{barrier constraint} associated with the CBF. According to the definition of the CBF, for all $x \in \safeset$, the barrier constraint must be feasible. We denote the left hand side of \eqref{eq:barrier-constraint}, $\mu(x, u)$, as the \textit{CBF virtual control input}.
\end{definition}

Note that $\mu(x, u)$ captures the slackness (if positive) or violation (if negative) of the barrier constraint for a chosen $u$ at a given state $x$. The constraint is saturated if $\mu = 0$.

Whereas $k$ was the feedback gain in the IO linearization, here, $k$ is the constraint coefficient vector associated with the CBF. Due to the condition on $k$ in Definition \ref{def:ho-cbf}, we can denote the negative real roots of the polynomial $s^r + k_r s^{r-1} + \cdots + k_1=(s+\gamma_1 )(s+\gamma_2 )\cdots (s + \gamma_r)$ as $-\gamma_1, \cdots, -\gamma_r$, where $\gamma_i > 0$ for $i = 1, \cdots, r$. Without loss of generality, we set $\gamma_r$ to
$\gamma_{\min}:=\min_i \gamma_i$ by changing the orders of $\gamma_i$'s. Now, consider
\vspace{-1em}
\begin{align}
\phi_1 & := h(x), \nonumber \\
\phi_2 & := \dot{\phi}_1 + \gamma_1 \phi_1 \;\;(\resizebox{.28\hsize}{!}{$\displaystyle= L_f h(x) + \gamma_1 h(x)$}), \nonumber \\
\phi_3 & := \dot{\phi}_2 + \gamma_2 \phi_2 \;\;(\resizebox{.61\hsize}{!}{$\displaystyle= L_f^{(2)} h(x) + (\gamma_1 + \gamma_2) L_f h(x) + \gamma_1 \gamma_2 h(x)$}), \nonumber \\
\vdots & \nonumber \\
\phi_r & := \dot{\phi}_{r-1} + \gamma_{r-1} \phi_{r-1},
\label{eq:cascade-equations}
\end{align}
and refer to $\phi:= [\phi_1, \cdots, \phi_r]^\top$ as the \textit{cascading constraint vector} of the CBF $h$. A linear transformation from the CBF output derivative vector $\xi$ to the cascading constraint vector $\phi$ exists \cite{xu2016control}, described by
\vspace{-0.25em}
\begin{equation}
 \phi = \XiToPsi \xi, \label{eq:xitophi}
\vspace{-0.5em}
\end{equation}
where \vspace{-0.5em}
\begin{equation}
\XiToPsi=\resizebox{.7\hsize}{!}{$\displaystyle\left[\begin{array}{cccccc}
1 & 0 & 0 & \ldots & 0 & 0 \\
\gamma_1 & 1 & 0 & \ldots & 0 & 0 \\
\gamma_1 \gamma_2 & \gamma_1+\gamma_2 & 1 & \ldots & 0 & 0 \\
\vdots & \vdots & \vdots & \vdots & \ddots & \vdots \\
\prod_i^{r-1} \gamma_i & \ldots & \ldots & \ldots & \sum_{i=1}^{r-1} \gamma_i & 1
\end{array}\right].$}
\label{eq:output2cascade}
\end{equation}
The virtual control input $\mu$ and the barrier constraint on $\mu$ become
\begin{equation}
\label{eq:mu-phi}
    \mu = \dot{\phi}_r + \gamma_r \phi_r \ge 0.
\end{equation}
As can be seen in \eqref{eq:cascade-equations} and \eqref{eq:mu-phi}, the values of $\gamma_i$ determine how aggressively the value of $h(x)$ can decrease. If the values are higher, the barrier constraint becomes less restrictive but more myopic. If the values are smaller, the barrier constraint becomes more restrictive, however, it becomes more precautious about decreasing $h$. % $\gamma_{\min}$, which determines the slowest decaying mode of $h$, will later have an impact on the boundedness property of the CBF-constrained dynamics.

The main result for safety is now presented as below: 

\begin{theorem}\label{thm:HOCBF}
(Safety guarantee of CBF-constrained controller \cite[Theorem 3]{xiao2021high}). For a given exponential CBF $h$ and any state $x \in \safeset$, define the set of control input satisfying the barrier constraint as
\vspace{-0.5em}
\begin{equation}
    \Pi(x) := \{ u \in \R^m |  \; \mu(x, u) \ge 0 \}.
\label{eq:cbf-constraint}
\end{equation}
Then, any Lipschitz continuous controller $\pi:  \R^n \rightarrow \R^m$ such that $\pi(x) \in \Pi(x) \; \forall x \in \safeset$ will render the positive orthant of $\phi$ forward invariant. In other words, for any initial state $x_0$ in
\begin{equation}
 S_{\phi} :=\{x\in\R^n | \phi(x) \in \R^r_{\ge0} \} \subset \safeset, \label{eq:positive_orthant}
\end{equation}
under the control signal $\ctrl(t) \equiv \pi(\traj(t))$, the resulting trajectory satisfies $\traj(t) \in S_{\phi}$, and thus $h(\traj(t)) \ge 0$, for all $t \ge 0$.
\end{theorem}

Since $h$ is nonnegative in the positive orthant of $\phi$, i.e., $S_{\phi} \subset \safeset$, Theorem \ref{thm:HOCBF} ensures that \eqref{eq:safety-spec} is satisfied for all $x_0 \in S_{\phi}$; thus, safety is guaranteed under the barrier constraint. To use this theorem to design a safety filter with the CBF, one has to ensure that whenever a desired control command does not satisfy the barrier constraint---i.e., if $u \notin \Pi(x)$---it is overridden with $u \in \Pi(x)$.

\vspace{0.5em}

\begin{remark}
Definition \ref{def:ho-cbf} introduces a special case of the \textit{high-order} CBF, which can be generalized by replacing $\gamma_i \phi_i$ in \eqref{eq:cascade-equations} with general class $\mathcal{K}$ functions \cite{xiao2021high}.
\end{remark}

\section{Single-Input Case}
\label{sec:minimum-phase-cbf}

Before we conduct our main analysis for the single-input case, we first review the original minimum phase condition and its implication to the state boundedness in IO linearization.

\subsection{Background: Minimum Phase Conditions in Single-Input IO Linearization}
\label{subsec:single-input-io}

% We also consider conditions under which we can find the internal state whose dynamics are not actuated by $u$. In order for $u$ to vanish in $\dot{\eta}$, we need the following condition:

% \begin{condition}[Non-actuated internal state]
% \label{cond:nonactuated_internal_state}
% \begin{equation}
%     L_{g_i} \eta_j(x) = 0
% \end{equation}
% for all $x$ in the open set around $x_0$, $i = 1, \cdots, m$, and $j =1, \cdots, n-r$.    
% \end{condition}

In the single-input case with $m=1$, the following holds.

\begin{theorem}
\label{thm:diffeomorphism}
(\!\!\cite[Proposition 5.1.2]{isidori2013nonlinear}) Let $h$ have relative degree $r$ at $x_0$. Then, there exists $n-r$ functions $\eta_1(x), \cdots, \eta_{n-r}(x)$ such that Condition \ref{cond:nonsingular} holds and there exists an open set around $x_0$ such that for all $x$ in the set, $L_{g} \eta_j(x) = 0$ for $j =1, \cdots, n-r$.
\end{theorem}

The second condition implies that $u$ vanishes in $\dot{\eta}$, thus, the internal dynamics is not actuated by $u$ with $\eta$ defined in Theorem \ref{thm:diffeomorphism}. Then, the dynamics under IO linearization can be represented as:

\begin{definition}[Normal form for IO linearization]
\label{def:io-normal-form}
\begin{subequations}
\begin{align}
    \dot{\xi} & = A_k \xi, \label{eq:normal-form-io-output}\\
    \dot{\eta} & = \qio (\eta, \xi), \label{eq:normal-form-io-internal}
\end{align}
\label{eq:normal-form-io}    
\vspace{-1em}
\end{subequations}

\noindent where \eqref{eq:normal-form-io-output} is the output dynamics, and \eqref{eq:normal-form-io-internal} is the internal dynamics where the control $u$ vanishes. The output derivative vector $\xi$ can be considered a virtual control input for the internal dynamics.
\end{definition}

Next, as the objective in IO linearization is to regulate $\xi$ to zero, we consider a special case of the internal dynamics when $\xi\equiv \xi_e := 0$ and define the following:

\begin{definition}
\label{def:io-zero-dynamics}
(Zero dynamics manifold \& zero dynamics \cite[Section 9.2.2]{sastry2013nonlinear}) The \textit{zero dynamics manifold} of the system \eqref{eq:system} for the output $y=\cbf(x)$ is given by
\begin{equation}
    \mathcal{M}_{\text{io}} := \{x \; | \; \xi \equiv \xi_e = 0 \},
\end{equation}
on which the output is held identically zero. This is a manifold of dimension $n-r$. The dynamics of the internal coordinates constrained on $\mathcal{M}_{\text{io}}$, is given as
\begin{equation}
    \dot{\eta} = \qio (\eta, 0),
    \label{eq:io-zero-dynamics}
\end{equation}
which is referred to as the \textit{zero dynamics} of the system \eqref{eq:system} for the output $y=\cbf(x)$.
\end{definition}

\begin{definition}
\label{def:min-phase-io}
(Minimum phase conditions for IO linearization \cite[Definition 9.10]{sastry2013nonlinear}). The system \eqref{eq:system} is said to satisfy the local (global) asymptotic (exponential) \textit{minimum phase condition} for the output $y\!=\!\cbf(x)$ at $x_0$ if an equilibrium point $\eta_e$ of the zero dynamics \eqref{eq:io-zero-dynamics} is locally (globally) asymptotically (exponentially) stable.    
\end{definition}

The main implication of the above minimum phase condition is as follows:

\begin{theorem} (\!\!\cite[Theorem 9.13]{sastry2013nonlinear}) If the system \eqref{eq:system} is locally exponentially minimum phase for the output $y=\cbf(x)$, then the IO linearization control law given by \eqref{eq:io-linearization-siso} and \eqref{eq:io-virtual-input-law}, with $k$ selected such that $A_k$ is Hurwitz, results in $(\xi_e, \eta_e)$ to be locally exponentially stable.
\label{thm:siso-min-phase}
\end{theorem}

In other words, Theorem \ref{thm:siso-min-phase} states that the minimum phase condition in Definition \ref{def:min-phase-io} is a sufficient condition for the coupled dynamics in \eqref{eq:normal-form-io} to be exponentially stable. This also implies that the system state will stay bounded under the IO linearization control law. In contrast, when the system is under a nonminimum phase condition (minimum phase conditions not satisfied), we do not have a guarantee that the system state will stay bounded (even when $\xi$ exponentially converges to 0). In many practical nonminimum phase examples, we can observe system states associated with the internal state exploding under the IO linearization control law.

\vspace{-0.5em}

\subsection{Two Representations of CBF-constrained Dynamics}

Next, returning to the analysis of the CBF-constrained dynamics, similarly to IO linearization, we present representations of the CBF-constrained dynamics that divide the dynamics into CBF output and internal dynamics. 
% and look into the features of the CBF output dynamics, which serve as preliminaries for the main analysis in the next section.
We will present two ways of representing the CBF-constrained dynamics. 
% First, we present an assumption that allows such representations:

% \begin{assumption}[Diffeomorphic change of coordinates]
% \label{assumpt:cbf-dffeomorphism}
% For the CBF $h$ and its \barrierconstraint, for all $x_0 \in S_{\phi}$, the diffeomorphic change of coordinates $\Psi_{\xi}$ is well defined based on Theorem \ref{thm:diffeomorphism}.
% \end{assumption}

% \jcnote{Assumption \ref{assumpt:cbf-dffeomorphism} can be satisfied if Assumption \ref{assump:involutive} is satisfied for all $x_0 \in S_{\phi}$.} Under this assumption, we can consider two representations of the original system dynamics under the barrier constraint.

\begin{definition}[Normal form for CBF-constrained dynamics] 
\label{def:cbf-normal-form}
In the \textit{normal form}, we describe the original system \eqref{eq:system} under the transformed coordinates $(\xi, \eta)$ as
\begin{subequations}
\begin{align}
    \dot{\xi} & = A_k \xi + B \mu  \label{eq:cbf-output-dynamics-normal}\\ 
    \dot{\eta} & = \pcbf (\eta, \xi), \label{eq:cbf-internal-dynamics-normal}
\end{align}
\label{eq:normal-form}
\end{subequations}
\noindent\!\!\!where $A_k$ is given by \eqref{eq:A_xi}, $B:=[0\; \cdots\; 0\; 1]^\top$,
and $\pcbf:\R^{n-r} \times \R^{r} \rightarrow \R^{n-r}$ is the internal dynamics. Here, $\eta$ must satisfy the conditions in Theorem \ref{thm:diffeomorphism}.
% \begin{equation}
%     \dot{\xi} = \underbrace{\begin{bmatrix} 0 & 1 & 0 & \cdots & 0 \\
%     0 & 0 & 1 & \cdots & 0 \\
%     \; & \; & \ddots & \; & \; \\
%     0 & 0 & 0 & \cdots & 1 \\    
%     \; & \; & -k^\top & \; &\; \end{bmatrix}}_{A_{\xi}} \xi + \underbrace{\begin{bmatrix} 0 \\ \vdots \\ 0 \\ 1\end{bmatrix}}_{B_{\xi}} \mu,
% \label{eq:normal-form-output}
% \end{equation}
% and 
% \begin{equation}
%     \dot{\eta} = q_{\xi}(\eta, \xi).
% \label{eq:output-form-zero}
% \end{equation}
\end{definition}

Equation \eqref{eq:cbf-output-dynamics-normal} results from checking that
\begin{equation*}
    \dot{\xi}_r = L_f^r \cbf(x) +  L_g L_f^{r - 1} \cbf(x) u = \mu(x, u) - k^\top \xi(x).
\end{equation*}
While the normal form for IO linearization in \eqref{eq:normal-form-io} defines the \textit{closed-loop} output dynamics under the linear feedback law in \eqref{eq:io-virtual-input-law}, the normal form in \eqref{eq:normal-form} for the CBF-constrained dynamics defines an \textit{open-loop} output dynamics with a restriction on the control, $\mu \ge 0$, from the barrier constraint. This results in the additional $B \mu$ term in \eqref{eq:cbf-output-dynamics-normal}, compared to \eqref{eq:normal-form-io-output}.

From the coordinate transformation defined in \eqref{eq:xitophi} between $\xi$ and $\phi$, we can also consider an equivalent representation of the systems under the coordinates $(\phi, \eta)$. Since $T$ in \eqref{eq:output2cascade} is invertible, if $\Psi_{\xi}: x \rightarrow [\xi^\top, \eta^\top ]^\top$ is a diffeomorphism, the coordinate transformation $\Psi_{\phi}:x \rightarrow [\phi^\top, \; \eta^\top]^\top$ is also a diffeomorphism. This leads to the second representation of the CBF-constrained dynamics, where the normal form in Definition \ref{def:cbf-normal-form} is transformed into the following form:

\begin{definition}[Cascading constraint form]
\begin{subequations}
\begin{align}
\dot{\phi} & = A_{\gamma} \phi + B \mu  \label{eq:cascade-form-output} \\
\dot{\eta} & = \qcbf (\eta, \phi), \label{eq:cascade-form-internal}
\end{align}
\label{eq:cascade-form}    
\end{subequations}
% \begin{equation}
%     \begin{bmatrix}\dot{\phi} \\ \dot{\eta} \end{bmatrix} = \begin{bmatrix}A_{\gamma} \phi + B_\phi \mu  \\ q_{\phi}(\eta, \phi) \end{bmatrix},
% \label{eq:cascade-form}    
% \end{equation}
where 
\vspace{-0.5em}
\begin{equation*}
    A_{\gamma}:= T A_{k} T^{-1} = \resizebox{0.45\hsize}{!}{$\displaystyle\begin{bmatrix} -\gamma_1 & 1 & 0  & 0 & \cdots & 0 \\
                                                    0 & -\gamma_2 & 1 & 0 & \cdots & 0 \\
                                                    \; & \; & \ddots & \; & \; & \; \\
                                                    0 & \cdots & \; & 0 & -\gamma_{r-1} & 1\\
                                                    0 &  \; & \cdots & \; & 0 & -\gamma_r 
                                                    \end{bmatrix}$},
\end{equation*}
and $B = \XiToPsi B = [0\;\cdots\; 0\; 1]^\top$ determine the dynamics of $\phi$, and $\qcbf(\eta, \phi):= \pcbf (\eta, \XiToPsi^{-1} \phi)$ is the internal dynamics, where $\phi$ can be considered as its virtual control input.
\end{definition}

We call \eqref{eq:cbf-output-dynamics-normal}, \eqref{eq:cascade-form-output} the \textit{CBF output dynamics} and \eqref{eq:cbf-internal-dynamics-normal}, \eqref{eq:cascade-form-internal} the \textit{CBF internal dynamics}. With the assumption that $h$ is Lipschitz continuous, as $f, g$ are Lipschitz, the internal dynamics $\qcbf$ is also Lipschitz in $\eta$, $\phi$. Its Lipschitz constants are denoted as $l_\eta$, $l_\phi$ with respect to $\eta$, $\phi$, respectively.

\vspace{-0.6em}

\subsection{Features of the CBF Output Dynamics}
\label{subsec:features}

We next examine the CBF output dynamics in \eqref{eq:cascade-form-output}. The eigenvalues of $A_k$ and $A_{\gamma}$ are $-\gamma_1, \cdots, -\gamma_r$, implying that both matrices are total negative, and thus Hurwitz and invertible. As such, $\gamma_{\min}$ becomes the slowest rate of convergence of the output dynamics. Next, define
\vspace{-0.35em}
\begin{equation}
    \Gamma := -A_{\gamma}^{-1}B =\resizebox{0.27\hsize}{!}{$\displaystyle\begin{bmatrix} (\gamma_1 \gamma_2 \gamma_3 \cdots \gamma_r) ^ {-1} \\ (\gamma_2 \gamma_3 \cdots \gamma_r)^{-1} \\ \vdots \\ \gamma_r^{-1} \end{bmatrix}$}.
    \label{eq:def-Gamma} \vspace{-0.35em}
\end{equation}
From \eqref{eq:cascade-form-output}, for any $\mu \in \R$, $\phi_e = \Gamma \mu$ is a stable controlled equilibrium of the CBF output dynamics by seeing that $\dot{\phi}|_{\phi_e} = A_{\gamma} \Gamma \mu + B \mu = 0$. Thus, \vspace{-0.35em}
\begin{equation}
\Phi_e:=\{\phi \;|\; \exists \mu \ge 0, \text{s.t.}\; \phi = \Gamma \mu\}
\label{eq:equilibria-manifold} \vspace{-0.35em}
\end{equation}
is a straight line in the space of $\phi$ that the output dynamics converges to. On this line, 
we can check that 
% \begin{equation}
% \phi_{k} = \gamma_{k-1} \phi_{k-1}, \quad \xi_{k} = 0.
% \end{equation}
$\phi_{k} = \gamma_{k-1} \phi_{k-1}$ for all $k=2, \cdots, r$.

Given an almost-everywhere differentiable virtual control input signal $\trajmu(t)$, consider \vspace{-0.35em}
\begin{equation}
    \trajphid(t) := \Gamma \trajmu(t), \vspace{-0.35em}
\end{equation}
a signal on $\Phi_e$ moving in time, and define \vspace{-0.35em}
\begin{equation}
    \trajphidiff(t) := \trajphi(t) - \trajphid(t), \vspace{-0.35em}
\end{equation}
the error between $\trajphi(t)$ and $\trajphid(t)$.
Then, we get \vspace{-0.35em}
\begin{align}
    \dot{\trajphidiff}(t) &= \dot{\trajphi}(t) - \dot{\trajphid}(t) = A_{\gamma} \trajphi(t) + B_\phi \trajmu(t) - \Gamma \dot{\trajmu}(t)  \nonumber \\
    & = A_{\gamma} (\trajphi(t) - \trajphid(t)) + \underbrace{(A_{\gamma} \trajphid(t) + B_\phi \trajmu(t))}_{=0} - \Gamma \dot{\trajmu}(t)  \nonumber \\
    & = A_{\gamma} \trajphidiff(t) - \Gamma \dot{\trajmu}(t),
\label{eq:error-dynamics}
\vspace{-0.35em}
\end{align}
which describes the \textit{error dynamics of} $\phi$. If $\trajmu(\cdot)$ is constant, the error $\trajphidiff$ converges exponentially to zero since $A_\gamma$ is Hurwitz. Moreover, we can prove that $\trajphidiff$ stays bounded when $\trajmu(\cdot)$ is Lipschitz continuous in time (Lemma \ref{lemma:boundedness-phidiff} in Appendix \ref{appen:lemmas}). 

Finally, a crucial observation that can be made in \eqref{eq:cascade-form-output} is the nonnegativity of $\mu$ and $\phi_i$'s under the \barrierconstraint. The nonnegativity of $\mu$ results directly from the barrier constraint. The nonnegativity of $\phi_i$'s results from Theorem \ref{thm:HOCBF}. 
% This observation is summarized as the following proposition:

% \begin{proposition} 
% \label{prop:nonnegativity-output}
% Under the dynamics subjected to the \barrierconstraint~\eqref{eq:barrier-constraint}, for any initial state $x_0 \in S_{\phi} \!=\!\{x\in\R^n \; | \;\phi(x) \in \R^r_{\ge0} \}$, we have $\trajmu(t)\ge0$ and $\trajphi(t) \in \R^r_{\ge0}$ for all $t \ge 0$.
% \end{proposition}

% \vspace{1em}

To summarize the key features of the CBF output dynamics, under the barrier constraint, the positive orthant of $\phi$, $S_\phi$, is invariant and contains a line consisting of stable equilibria $\Phi_e$.

\vspace{-0.6em}

\subsection{CBF Minimum Phase Conditions}

\label{subsec:minimum-phase}

Now, using the cascading constraint form of the CBF-constrained dynamics, we examine the CBF internal dynamics in \eqref{eq:cascade-form-internal}. If we consider $\phi$ as a virtual control input, the internal dynamics is driven only by a nonnegative control signal. This is the key aspect that differentiates the analysis of the CBF-constrained dynamics from the analysis in Section \ref{subsec:single-input-io}.
%, where a similar normal form is defined but there is no nonnegativity constraint in regulating both the output and the internal dynamics. 

% The restriction that the internal dynamics are regulated only by a nonnegative control signal motivates us to investigate under which condition the internal state can stay bounded and not ``blow up''. Such conditions in IO linearization is provided by the minimum phase conditions, for instance, as in Theorem \ref{thm:siso-min-phase}. Here, similarly, we provide a set of minimum phase conditions for the CBF internal dynamics, together with the associated theorem that proves the boundedness of the state trajectory.

We start with a simple setting that is most similar to the minimum phase conditions in the IO linearization. Recall that Definition \ref{def:min-phase-io} is defined as the condition of zero dynamics having a stable equilibrium. The ``zero'' of the zero dynamics in \eqref{eq:io-zero-dynamics} is in fact the equilibrium of the output dynamics, $\xi_e = 0$, that the output derivative vector $\xi$ converges to. 

In the CBF-constrained dynamics, the cascading constraint vector $\phi$ in the CBF output dynamics \eqref{eq:cascade-form-output} does not have to converge to zero. In fact, if $\phi$ converges to 0, it means the system is on the verge of violating safety as $h(x) = 0$, which is undesirable. Instead, as discussed in Section \ref{subsec:features}, we have a continuum of equilibrium on the line $\Phi_e$. Thus, the first case we study is when $\mu$ is fixed to a certain value, $\mu_e\ge0$, which results in the output coordinates $\phi$ to converge to $\phi_e = \Gamma \mu_e$. We define our first minimum phase condition for this case, similarly to Definition \ref{def:min-phase-io} as below:

\begin{definition}[Fixed-virtual-control-input exponential minimum phase condition] 
\label{cond:minimum-phase-simple}
If there exists $\mu_e \ge 0$, such that there exists a locally (globally) exponentially stable equilibrium, $\eta_e$, of the internal dynamics under the fixed output derivative vector value, $\trajphi(t)\!\equiv\!\Gamma \mu_e=\phi_e$, i.e., 
    \begin{equation}
        \dot{\eta} = \qcbf (\eta, \Gamma \mu_e),
        \label{eq:cbf-internal-dynamics-simple}
    \end{equation}    
we say that the CBF $h$ satisfies the \textit{fixed-virtual-control-input local (global) exponential minimum phase condition}.
\end{definition}

% Next, we show that the minimum phase condition defined in Condition \ref{cond:minimum-phase-simple} is a sufficient condition to render the CBF-constrained dynamics to be bounded.

With this CBF minimum phase condition defined, the first result we present is the implication of the local exponential minimum phase condition on the CBF-constrained dynamics:

\begin{theorem}
\label{thm:min-phase-simple-local} (Sufficient condition for local exponential stability of CBF-constrained dynamics)
For the CBF $h$ and the initial state $x_0 \in S_{\phi}$,
if the fixed-virtual-control-input \textit{local} exponential minimum phase condition (in Definition \ref{cond:minimum-phase-simple}) is satisfied,
then, under $\trajmu(t)\equiv \mu_e$, the dynamics is \textit{locally} exponentially stable at $(\phi, \eta) = (\Gamma \mu_e, \eta_e)$, where $\mu_e, \eta_e$ are defined in Definition \ref{cond:minimum-phase-simple}.
\end{theorem}

\begin{proof}
The proof resembles the proof of Theorem \ref{thm:siso-min-phase} in \cite{sastry2013nonlinear}. Under $\trajmu(t)\equiv \mu_e$, $\dot{\trajphidiff}(t) = A_{\gamma} \trajphidiff(t)$ from \eqref{eq:error-dynamics}. Thus, the closed-loop system of ($\phidiff, \eta$) becomes
\begin{align}
\label{eq:fixed-closed-loop}
    \dot{\phidiff} & = A_{\gamma} \phidiff \\
    \dot{\eta} & = \qcbf (\eta, \Gamma \mu_e + \phidiff) \nonumber
\end{align}
The linearization of the above system at ($\phidiff, \eta$) = ($0, \eta_e$) becomes \vspace{-0.5em}
\begin{equation*}
    \left[\begin{array}{cc}
A_{\gamma} & 0 \\
\frac{\partial \qcbf}{\partial \phi}(\eta_e,\Gamma \mu_e) & \frac{\partial \qcbf}{\partial \eta}(\eta_e,\Gamma \mu_e)
\end{array}\right].
\end{equation*}
Note that $A_{\gamma}$ is Hurwitz, and from Definition \ref{cond:minimum-phase-simple}, $\frac{\partial \qcbf}{\partial \eta}(\eta_e,\Gamma \mu_e)$ is also Hurwitz. Thus, by Lyapunov indirect theorem, ($0, \eta_e$) is a locally exponential equilibrium of the dynamics of ($\phidiff, \eta$), while the barrier constraint is satisfied by $\trajmu(t)\equiv \mu_e \ge 0$.
\end{proof}

% \begin{remark} 
% \label{remark:comparison} (Comparison to IO linearization)
% Definition \ref{cond:minimum-phase-simple} closely resembles Definition \ref{def:min-phase-io} since we first consider the notion of minimum phase closest to the existing IO linearization literature. In fact, Definition \ref{def:min-phase-io} for IO linearization is analogous to a special case of Definition \ref{cond:minimum-phase-simple} when $\mu_e = 0$, and Theorem \ref{thm:min-phase-simple-local} is analogous to Theorem \ref{thm:siso-min-phase}. % Thus, for any $h$ that can be IO linearized and stabilized to $0$ (by satisfying Definition \ref{def:min-phase-io}), if $h$ were the CBF, the CBF-constrained dynamics could also stay bounded (by saturating the barrier constraint, $\mu = 0$).
% \end{remark}

\begin{remark}[Comparison to IO linearization]
\label{remark:comparison} Definition \ref{def:min-phase-io} for IO linearization is analogous to a special case of Definition \ref{cond:minimum-phase-simple} when $\mu_e = 0$, and Theorem \ref{thm:min-phase-simple-local} is analogous to Theorem \ref{thm:siso-min-phase}.
\end{remark}

Theorem \ref{thm:min-phase-simple-local} is not sufficient for the boundedness of the CBF-constrained dynamics since it only provides the \textit{local} stability condition. If the initial condition $x_0$ is outside the region of attraction of $(\Gamma \mu_e, \eta_e)$, the state might still diverge from the equilibrium and go unbounded. Next, we show that if $h$ is Lipschitz continuous, with the Lipschitzness of the internal dynamics $\qcbf$, the \textit{global} minimum phase condition becomes a sufficient condition for the boundedness of the CBF-constrained dynamics.

\begin{theorem}
\label{thm:min-phase-simple-global} (Sufficient condition for boundedness of CBF-constrained dynamics---fixed-virtual-control-input case)
For the CBF $h$ which is Lipschitz continuous and the initial state $x_0 \in S_{\phi}$, if the fixed-virtual-control-input \textit{global} exponential minimum phase condition (in Definition \ref{cond:minimum-phase-simple}) is satisfied,
then, under $\trajmu(t)\equiv \mu_e$, the system converges to $\left(\trajphi(t), \trajeta(t) \right) \rightarrow \left(\Gamma \mu_e, \eta_e \right)$. Thus, the system under the \barrierconstraint~can stay bounded.
\end{theorem}

\begin{proof}
See Appendix \ref{appendix:simple-global-proof}.
\end{proof}

Theorem \ref{thm:min-phase-simple-global} is the first establishment that the boundedness of the CBF-constrained dynamics can be ensured when the minimum phase condition is secured. However, its implication is still limited as the boundedness is guaranteed by setting $\trajmu(\cdot)$ to a fixed value, constraining the use of CBFs for safety filtering. Thus, we next explore a more relaxed minimum phase condition that still ensures the boundedness of the state trajectory under the barrier constraint. It can be achieved by considering an internal dynamics that may be stabilized through a feedback control law:

\begin{definition}[Positive-feedback-stabilizable exponential minimum phase condition] 
\label{cond:minimum-phase-feedback} If there exists a Lipschitz feedback controller $\pspi:\R^{n-r}\rightarrow \R_{\ge0}$, which determines $\mu$, such that the internal dynamics under $\trajphi(t)= \Gamma \pspi(\trajeta(t))$, i.e.,
\begin{equation}
    \dot{\eta} = \qcbf\left(\eta, \Gamma \pspi(\eta)\right),
    \label{eq:cbf-internal-dynamics-feedback}
\end{equation}
has a locally (globally) exponentially stable equilibrium $\eta_e$, 
we say that the CBF $h$ satisfies the \textit{positive-feedback-stabilizable local (global) exponential minimum phase condition}.
\end{definition}

% In words, we say that the CBF-constrained dynamics satisfy the minimum phase condition, when the internal dynamics is positive-feedback-stabilizable. 
Note that Definition \ref{cond:minimum-phase-simple} now becomes the special case of Definition \ref{cond:minimum-phase-feedback} when $\pspi(\eta) \equiv \mu_e$. We now show that under the new global minimum phase condition, with some additional conditions, the CBF-constrained dynamics can stay bounded.

\begin{figure}[t]
\centering
\includegraphics[width=\columnwidth]{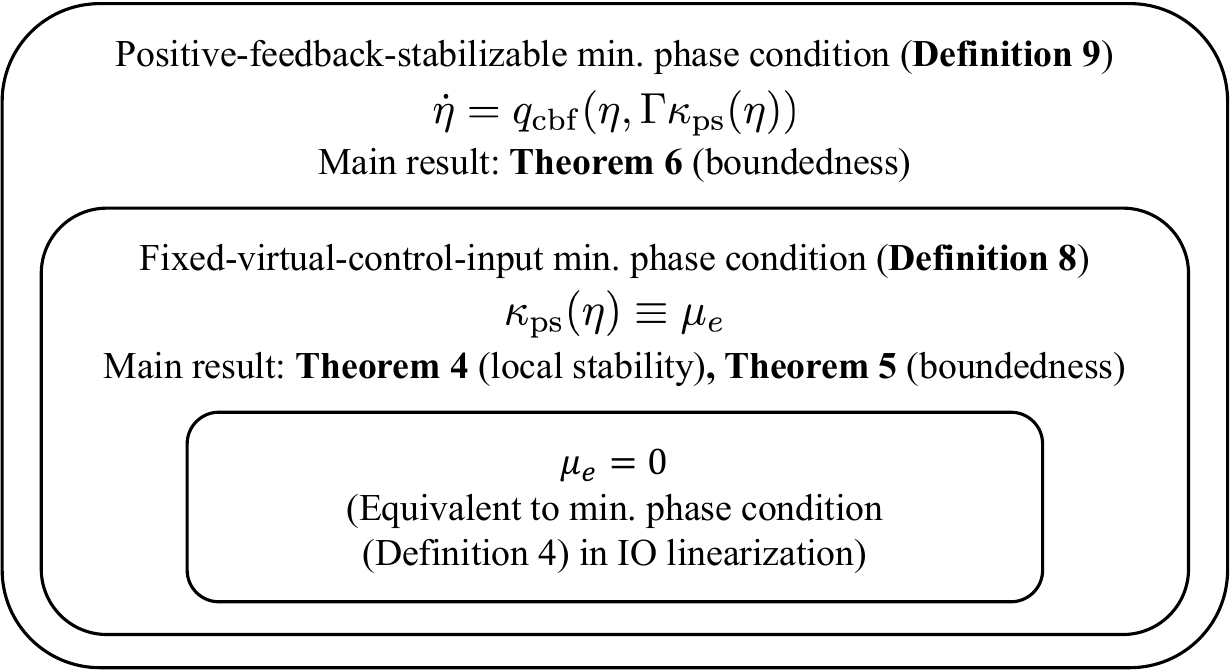}
\caption{Relationship between different CBF minimum conditions. For linear systems, all conditions collapse to the same condition that $A_\eta$ in \eqref{eq:linear-case-zero} itself has to be Hurwitz.}
\vspace{-0.5em}
\label{fig:venn-diagram}
\end{figure}

\begin{table}[]
\caption{Comparison between the analysis for the dynamics under the input-output (IO) linearization and the CBF filter. In each case, $h$ is the output and the CBF, respectively.}
\label{table:main-comparison}
\renewcommand{\arraystretch}{1.5} % Default value: 1
\addtolength{\tabcolsep}{-4pt}
\begin{tabular}{|l|c|c|}
\hline
\; & IO linearization \cite{sastry2013nonlinear} & CBF-filtered (simple ver.) \\ \hline \hline
Virtual ctrl input & \begin{tabular}[c]{@{}c@{}} {\scriptsize $\nu:= \cbf^{(r)} (x, u)$}\end{tabular} & \begin{tabular}[c]{@{}c@{}}{\scriptsize$\mu:=\cbf^{(r)}(x, u) \!+\! k^\top \xi(x)$}\end{tabular} \\ \hline
Control law                                                           & \begin{tabular}[c]{@{}c@{}}$\nu = - k^\top \xi(x)$ \\ (output-zeroing-feedback)\end{tabular}    & \begin{tabular}[c]{@{}c@{}}$\mu\ge 0$ \eqref{eq:barrier-constraint}\\ (barrier constraint)\end{tabular}    \\ \hline
Output dynamics & \begin{tabular}[c]{@{}c@{}}$\dot{\xi} = A_k \xi$ \eqref{eq:io-output-dynamic-closed} \\ (closed-loop)\end{tabular}    & \begin{tabular}[c]{@{}c@{}}
$\dot{\xi} = A_k \xi + B \mu$ 
\eqref{eq:cbf-output-dynamics-normal} $/$ \\
$\dot{\phi} = A_{\gamma} \phi + B \mu$ \eqref{eq:cascade-form-output} \\ where $\phi=T \xi$ \eqref{eq:xitophi}\end{tabular}    \\ \hline
\begin{tabular}[c]{@{}l@{}}Property of the \\ system matrix\end{tabular} & $A_k$ Hurwitz    & \begin{tabular}[c]{@{}c@{}} $A_k \;/\; A_\gamma$ Total negative % \\  ($-\gamma_1, \cdots, -\gamma_r$)
\end{tabular}    \\ \hline
\begin{tabular}[c]{@{}l@{}} Equilibria in \\ output space \end{tabular} & $\xi_e=0$, the stable eq.& \begin{tabular}[c]{@{}c@{}} $\phi_e \!=\!\Gamma\mu$, $\forall \mu \ge 0$, \\ is a stable controlled eq. \end{tabular}    \\ \hline

\begin{tabular}[c]{@{}l@{}}Forward invariance \\ for safety \end{tabular} & N\slash A
&
$\big\{\phi \; | \; \phi \in \R^\gamma_{\ge0}\big\}$  \\ \hline

\begin{tabular}[c]{@{}l@{}}Diagram in \\ output space \end{tabular} & \begin{tabular}[c]{@{}l@{}}\; \vspace{-0.5em} \\ \includegraphics[width=0.3\columnwidth]{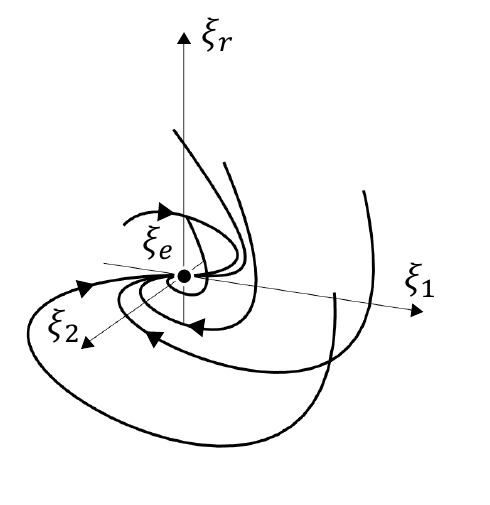}\end{tabular} &  \begin{tabular}[c]{@{}l@{}}\; \vspace{-0.5em} \\ \includegraphics[width=0.35\columnwidth]{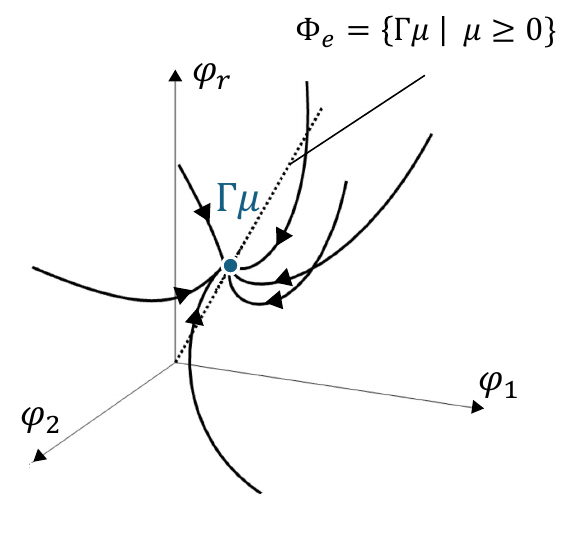}\end{tabular} \\ \hline

Internal dynamics & $\dot{\eta} = \qio (\eta, \xi)$ \eqref{eq:normal-form-io-internal} & $\dot{\eta} = \qcbf (\eta, \phi)$ \eqref{eq:cascade-form-internal}\\ \hline
% \begin{tabular}[c]{@{}l@{}}Zero dynamics \\ manifold\end{tabular} & $\{x \; | \; \xi 
% - \xi_e \equiv 0 \}$ & \begin{tabular}[c]{@{}c@{}}$\{x \; | \; \phi - \Gamma \mu_e \equiv 0\}$\\ where $\mu_e$ defined in Cond. \ref{cond:minimum-phase-simple}\end{tabular} \\ \hline

Zero dynamics & $\dot{\eta} = \qio (\eta, 0)$ \eqref{eq:io-zero-dynamics} & $\dot{\eta} = \qcbf (\eta, \Gamma \mu)$ \eqref{eq:cbf-zero-dynamics} \\ \hline

% \begin{tabular}[c]{@{}l@{}}Min-phase cond.\end{tabular} & Condition \ref{def:min-phase-io} & Condition \ref{cond:minimum-phase-simple} \\ \hline

\begin{tabular}[c]{@{}l@{}}Minimum phase \\ condition\end{tabular} & \begin{tabular}[c]{@{}c@{}} zero dynamics \\ with stable equilibrium \end{tabular} & \begin{tabular}[c]{@{}c@{}} positive stabilizable \\ zero dynamics \end{tabular} \\ \hline

% \begin{tabular}[c]{@{}l@{}}Implication of \\ Min-phase cond.\end{tabular} & Theorem \ref{thm:siso-min-phase} & Theorems \ref{thm:min-phase-simple-local} \& \ref{thm:min-phase-simple-global} \\ \hline

\begin{tabular}[c]{@{}l@{}}Implication of \\ Min. phase cond.\end{tabular} & \begin{tabular}[c]{@{}c@{}} Full state stability \\ (Theorem \ref{thm:siso-min-phase}) \end{tabular} & 

\begin{tabular}[c]{@{}c@{}} Full state boundedness \\ (Theorem \ref{thm:min-phase-feedback})   \end{tabular} \\ \hline
\end{tabular}
\vspace{-1em}
\end{table}

\begin{theorem}
\label{thm:min-phase-feedback} 
(Sufficient condition for boundedness of CBF-constrained dynamics---positive-feedback-stabilizable case)
For the CBF $h$ which is Lipschitz continuous and the initial state $x_0 \in S_{\phi}$, let the positive-feedback-stabilizable global exponential minimum phase condition (in Definition \ref{cond:minimum-phase-feedback}) hold. Without loss of generality, we set $\eta_e=0$. From Definition \ref{cond:minimum-phase-feedback}, by the converse Lyapunov theorem, there exists $V$ such that \vspace{-0.25em}
\begin{equation}
\begin{aligned}
        \alpha_1 |\eta|^2 \le V(\eta) \le \alpha_2 |\eta|^2, \quad \left|\frac{\partial V}{\partial \eta}\right| \le \alpha_3 |\eta|, \\
    \frac{\partial V}{\partial \eta} \cdot \qcbf (\eta, \Gamma \pspi(\eta)) \le - \alpha_4 |\eta|^2,
\end{aligned}
\label{eq:min-phase-lyap} \vspace{-0.5em}
\end{equation}
with $\alpha_{1, 2, 3, 4} \!>\!0$. If $\alpha_4\!>\!\left(\!{\alpha_3 l_\phi \over 2} \!\right)^2$ and if $\gamma_{\min}$ is large enough, the system under the \barrierconstraint~can stay bounded. In particular, by setting $\trajmu(t)=\pspi(\trajeta(t))$, the system converges to $\left(\trajphi(t), \trajeta(t)\right) \rightarrow \left(\Gamma \pspi(\eta_e), \eta_e \right)$, and thus stays bounded.
\end{theorem}
\vspace{-0.5em}
\begin{proof}
See Appendix \ref{appendix:feedback-theorem-proof}.\vspace{-0.25em}
\end{proof}

In the theorem, we have two important conditions in addition to the global minimum phase condition. First, the internal dynamics needs strong enough stability by virtue of the condition on $\alpha_4$. More importantly, as can be seen in the theorem proof (equation \eqref{eq:proof-gamma-poly}), $\gamma_{\min}$ has to be large enough as compared to the internal dynamics response. This requirement mainly results from the fact that the internal dynamics is not directly controlled by $\trajmu(t)$; it is controlled by $\trajphi(t)$. Thus, if we wanted the internal dynamics to behave like \eqref{eq:cbf-internal-dynamics-feedback} and be stabilized, we need $\trajphi(t)$ to quickly converge to $\Gamma \pspi(\trajeta(t))$. The transient behavior during convergence, governed by the output dynamics response \eqref{eq:cascade-form-output}, is dominated by the slowest decay rate $\gamma_{\min}$. In Section \ref{subsec:nonlinear-ex}, we highlight the necessity of the condition on $\gamma_{\min}$ through a nonlinear system example.

Note that Theorem \ref{thm:min-phase-feedback} is the possibility result for the boundedness. Similarly to Theorem \ref{thm:min-phase-simple-global}, Theorem \ref{thm:min-phase-feedback} is proved by showing that by setting $\mu$ to $\pspi$. However, it does not guarantee that \textit{any} nonnegative $\trajmu(\cdot)$ will result in the bounded state, even when the global minimum phase condition is satisfied. In contrast to the result in IO linearization, where any stable feedback law ensures the boundedness of the dynamics under the original minimum phase condition, such a universal result is not achievable for safety filters, as their closed-loop behavior depends on an arbitrary reference control command. 

\begin{remark}
Definition \ref{cond:minimum-phase-feedback} is framed in terms of positive stabilizability of $\eta$ and a constructive approach to verifying this property must be further developed, for instance, by verifying the Lyapunov function $V$ together with $\pspi$, satisfying \eqref{eq:min-phase-lyap}. We discuss the special case of linear systems in Remark \ref{remark:linear-system-case} and leave the more general case as a future research direction.
\end{remark}

\vspace{-1.1em}
\subsection{CBF Zero Dynamics}

Notice that, in our analysis so far, the CBF-constrained dynamics stays bounded when it can get stabilized to the dynamics constrained by $\phi \in \Phi_e$ and $\dot{\eta} = \qcbf\left(\eta, \Gamma \mu \right)$. As noted before, the zero dynamics in IO linearization is defined as the dynamics when the output derivative vector $\xi$ is constrained to the equilibrium zero. Similarly, we can define the zero dynamics of the CBF as the special case of the internal dynamics when the cascading constraint vector $\phi$ is constrained on the set of equilibria:

\begin{definition}[CBF zero dynamics \& zero dynamics manifold]
\label{def:zero-dynamics-cbf}
For the CBF $h$, the  \textit{CBF zero dynamics} of the system \eqref{eq:system} is defined as
\vspace{-0.5em}
\begin{equation}
    \dot{\eta} = \qcbf\left(\eta, \Gamma \mu \right),
    \label{eq:cbf-zero-dynamics}
\end{equation}
in which $\mu \ge 0 $ is considered as a virtual nonnegative input to the dynamics. This is a special case of the internal dynamics \eqref{eq:cascade-form-internal} where $\phi$ is constrained on the equilibria manifold $\Phi_e$ \eqref{eq:equilibria-manifold}. We define $\mathcal{M}_{\text{cbf}} := \{x \; | \; \phi = \Gamma \mu\}$ as the corresponding \textit{CBF zero dynamics manifold}.
\end{definition}

Whereas the zero dynamics \eqref{eq:io-zero-dynamics} in IO linearization does not involve any virtual control, the CBF zero dynamics \eqref{eq:cbf-zero-dynamics} includes a (nonnegative) scalar input $\mu$. Although no zeros define the zero dynamics for the CBF, we retain the naming convention from IO linearization to maintain consistency between the two bodies of literature. Using the new CBF zero dynamics definition, we can reiterate our previous definitions of the minimum phase conditions:
\begin{enumerate}[leftmargin=1.25em, labelindent=\parindent, listparindent=\parindent, labelwidth=0pt]
    \item (Definition \ref{cond:minimum-phase-simple}) The CBF satisfies the \textit{fixed-virtual-control-input minimum phase condition} if the CBF zero dynamics is positively stabilizable under a fixed virtual input $\mu_e \ge 0$.
    \item (Definition \ref{cond:minimum-phase-feedback}) The CBF satisfies the \textit{positive-feedback-stabilizable minimum phase condition} if the CBF zero dynamics is positively stabilizable under a nonnegative feedback control law $\pspi$.
\end{enumerate}
The new definitions presented for the CBF-constrained dynamics along with the analysis conducted so far are summarized in Table \ref{table:main-comparison} and Figure \ref{fig:venn-diagram} in comparison to IO linearization.

\vspace{1em}

\begin{remark}[Special case: linear systems]
\label{remark:linear-system-case}
The fact that the CBF zero dynamics form a single-input system (where the input is $\mu$) leads to interesting properties in the special case of linear systems. For linear systems, the internal dynamics \eqref{eq:cascade-form-internal} becomes $\dot{\eta} = A \eta + B \phi$, where $A \in \R^{(n-r)\times (n-r)}$, and $B \in \R^{(n-r) \times (r)}$. The CBF zero dynamics becomes
\begin{equation}
    \dot{\eta} = A \eta + B \Gamma \mu.
\label{eq:linear-case-zero}
\end{equation}
Revisiting the sufficient and necessary condition for the positive stabilizability of linear systems \cite{heemels1998positive}, we have an important conclusion for the linear system case, stated in Proposition \ref{cor:linear-system-case} below. It collapses the CBF minimum phase conditions to the minimum phase condition in IO linearization. Therefore, for single-input linear systems, we can guarantee the boundedness of the CBF-constrained dynamics only when the internal dynamics is open-loop stable.
\end{remark}

% \begin{definition}[Positive Stabilizability of Linear Systems \cite{heemels1998positive}] A linear system ($A, B$) is said to be \textit{positively stabilizable} if for every initial state $x_0 \in \R^n$, there exists an asymptotically stabilizing control signal $\ctrl(\cdot) \in \pcfset$ \footnote{For a linear system, a control signal $\ctrl \in \cfset$ is said to be asymptotically stabilizing for an initial state $x_0 \in \R^n$, if $\traj(t)\rightarrow 0$ as $t \rightarrow \infty$.}.
% \end{definition}

% \begin{theorem}\label{thm:positive_stabilizable}
% (\!\!\cite[Theorem 3.1]{heemels1998positive}). The linear system ($A, B$) is positively stabilizable if and only if $(A, B)$ is stabilizable and any real eigenvector of $A^\top$ corresponding to a positive eigenvalue of $A^\top$, $v \in \R^n$, satisfies that $B^\top v \in \R^m$ has at least one strictly positive element.
% \end{theorem}

% \noindent An important corollary of Theorem \ref{thm:positive_stabilizable} is as below:

% \begin{corollary}\label{thm:positive_stabilizable-special-case}
% (\!\!\cite[Corollary 3.3]{heemels1998positive}). If $m = 1$ (scalar input), the linear system ($A, B$) is positively stabilizable if and only if $A$ is Hurwitz.
% \end{corollary}

% As the CBF zero dynamics \eqref{eq:linear-case-zero} has a scalar input $\mu$, from Corollary \ref{thm:positive_stabilizable-special-case}, we have an important conclusion for the linear system case:
% \end{remark}

\vspace{1em}

\begin{proposition}
\label{cor:linear-system-case}
For single-input linear systems, the CBF satisfies the minimum phase condition if and only if $A$ in \eqref{eq:linear-case-zero} is Hurwitz. 
\end{proposition}

\begin{proof} This results from \cite[Corollary 3.3]{heemels1998positive}, which states that a single-input linear system $(A, B)$ is positively stabilizable if and only if $A$ is Hurwitz.
\end{proof}

\vspace{1em}

\subsection{Remark on the Feasibility of CLF-CBF-QP}
We conclude this section by establishing the relevance of our result to control designs that enforce both safety and stability. The exposition here also applies to the multi-input case discussed in Section \ref{sec:multi-input}.

One approach to guarantee boundedness of the system state is to stabilize the system state to an equilibrium using a control Lyapunov function (CLF) \cite{SONTAG1989117}. If the control input can be constrained to satisfy a CLF-based stability condition, this condition can be enforced alongside a CBF-based safety constraint to yield a control law that both safely and stably regulates the system. For instance, the work in \cite{Ames2016} enforces both conditions via the following quadratic program:

\newpage
\hrule
\vspace{2mm}
\noindent \textbf{CLF-CBF-QP \cite{Ames2016}}:
\begin{subequations}
\label{eq:clf-cbf-qp}
\begin{align}
\pi_{\text{cbf}}(x) & = & & \underset{u\in \R^{m}}{\arg\min}  \quad \norm{u-\pi_{\text{ref}}(x)}_2^2 \label{eq:clf-cbf-qp-cost}\\
& \text{s.t.} & & L_f W(x) + L_g W(x)u + \lambda  W(x) \le 0,\label{eq:clf-cbf-qp-clf} \\
& \; & & L_f^r \cbf(x) +  L_g L_f^{r - 1} \cbf(x) u + k^\top \xi(x) \ge 0. \label{eq:clf-cbf-qp-cbf}
\vspace{-.5em}
\end{align}
\end{subequations}
\hrule
\vspace{2mm}
\noindent Here, \eqref{eq:clf-cbf-qp-clf} encodes the CLF stability constraint on the control input $u$, where $W(x)$ denotes the CLF. However, in general, the feasibility of this QP is not guaranteed. The infeasibility arises when the CLF constraint \eqref{eq:clf-cbf-qp-clf} and the CBF constraint \eqref{eq:clf-cbf-qp-cbf} cannot be simultaneously satisfied. In such cases, one of the constraints must be relaxed to ensure feasibility. Often, the CLF condition is relaxed to prioritize safety, as in \cite{Ames2016}, resulting in a controller that may no longer ensure boundedness of the system state.

A similar issue arises in predictive safety control schemes where both safety and stability are simultaneously enforced \cite{wabersich2021predictive, Wabersich2023}. The work in \cite{didier2024predictive} proves the recursive feasibility of a predictive stability filter under the assumption that, at every point within the safe set, there exists a control input that simultaneously satisfies both the stability condition and the control invariance condition.

Thus, enforcing both safety and stability constraints requires special care to avoid conflicts between them. One solution is to design a CLF that is compatible with the given CBF \cite{dai2024verification}. Alternatively, a single function satisfying both objectives can be constructed from scratch, such as in the control Lyapunov-barrier function framework \cite{romdlony2016stabilization}.

In our work, Definition \ref{cond:minimum-phase-feedback} provides sufficient conditions for the feasibility of the CLF-CBF-QP. In particular, the Lyapunov function $W$ constructed in the proof of Theorem~\ref{thm:min-phase-feedback} serves as a CLF that is inherently compatible with the CBF.

\begin{corollary}[Guaranteed Feasibility of CLF-CBF-QP under the CBF minimum phase condition] 
Under the global exponential minimum phase condition and other assumptions of Theorem \ref{thm:min-phase-feedback}, with the CLF $W$ defined in equation \eqref{eq:lyapunov-total}, and with $\lambda > 0$ chosen as the minimum eigenvalue of the positive definite matrix appearing on the right-hand side of \eqref{eq:proof-uvdot}, the CLF-CBF-QP in \eqref{eq:clf-cbf-qp} is always feasible.
\end{corollary}

\begin{proof} This result follows directly from Theorem \ref{thm:min-phase-feedback}. When the control input $u$ satisfies $\mu(x, u) = \pspi(\eta(x))$, the stability constraint \eqref{eq:clf-cbf-qp-clf} is satisfied by \eqref{eq:proof-uvdot} and the barrier constraint \eqref{eq:clf-cbf-qp-cbf} is satisfied because $\mu$ is nonnegative.
\end{proof}

\section{Multi-Input Case}
\label{sec:multi-input}

Next, we extend our analysis to the multi-input case of $m \ge 2$. The main distinction from the single-input case is that the control input appears in the internal dynamics, as we show below. This contrasts with the single-input case, where we can always find the internal state coordinates $\eta$ in which $u$ disappears in its dynamics (Theorem \ref{thm:diffeomorphism}). As such, we will define the notion of the minimum phase condition in a closed-loop sense for the multi-input case. % In Sections \ref{subsec:multi-input-linear} and \ref{subsec:multi-input-nonlinear}, through multi-input linear and nonlinear examples, we will see that we can use the control terms that appear in the internal dynamics to actively stabilize the internal states while satisfying the barrier constraint.

% \subsection{Internal Dynamics \& CBF Minimum Phase Condition}
% \label{subsec:multi-input-analysis}

In order for the internal dynamics to \textit{not} depend on $u$, $\eta$ must satisfy the condition below:
\begin{condition}[Non-actuated internal state]
\label{cond:nonactuated_internal_state}
There exists an open set around $x_0$ such that for all $x$ in the set,
\begin{equation}
    L_{g_i} \eta_j(x) = 0
\end{equation}
for all $i = 1, \cdots, m$, and $j =1, \cdots, n-r$.    
\end{condition}
However, in the multi-input case with $g_1, \cdots, g_m$ linearly independent at $x$, the control term will \textit{always} show up in the internal dynamics by the following theorem.

\begin{theorem}
\label{thm:multi-input-internal}
% Under Assumption \ref{assmp:g_indepent}, 
If $g_1, \cdots, g_m$ are linearly independent at $x$, for any coordinate transformation $\eta:\R^{n} \rightarrow R^{n-r}$ differentiable at $x$, satisfying both Conditions \ref{cond:nonsingular} and \ref{cond:nonactuated_internal_state} is not possible.
\end{theorem}

\begin{proof}
    See Appendix \ref{appen:theorem-multi}.
\end{proof}

\begin{remark}
In \cite[Remark 5.1.3]{isidori2013nonlinear}, it is noted that finding the internal state coordinates $\eta$ whose dynamics does not depend on $u$ is possible for systems with different numbers of inputs and outputs when the distribution $\{g_1, \cdots, g_m\}$ is involutive. However, here, we prove that this is \textit{not} true for multi-input-single-output systems when $g_1, \cdots, g_m$ are linearly independent, regardless of the involutivity condition.
\end{remark}

Theorem \ref{thm:multi-input-internal} implies that with $\eta$ satisfying Condition \ref{cond:nonsingular}, establishing a valid diffeomorphism from $x$ to $(\phi, \eta)$, Condition \ref{cond:nonactuated_internal_state} cannot be satisfied and the CBF internal dynamics will be represented as
$$
\dot{\eta} = \qcbf (\eta, \phi, u),
$$
in which the control term $u$ always appears. Similarly to \cite{sastry1988feedback}, we consider the closed-loop dynamics of $\qcbf$ under a state-feedback control law $\pi:\R^{n}\rightarrow \R^m$ for $u$:
$$
\qcbf^{\pi}(\eta, \phi) := \qcbf (\eta, \phi, \pi(x)),
$$
where $\pi(x)$ can be expressed in terms of $(\phi, \eta)$ through the diffeomorphic transformation. By considering this closed-loop dynamics, we extend Definition \ref{cond:minimum-phase-feedback} to the multi-input case:

\begin{definition}[Closed-loop minimum phase condition for multi-input case] 
\label{cond:minimum-phase-multi-input} Consider a Lipschitz feedback controller $\pi:\R^{n}\rightarrow \R^m$, which determines $u$, satisfying the barrier constraint $\mu(x, \pi(x)) \ge 0$ for all $x \in S_\phi$. If under the policy $\pi$ and $\phi = \Gamma \mu$, the closed-loop internal dynamics $\qcbf^{\pi}$
% is stable,
% i.e.,
% \begin{equation}
%     \dot{\eta} = \qcbf^\pi (\eta, \Gamma \mu),
%     \label{eq:cbf-internal-dynamics-mimo-feedback}
% \end{equation}
has a (global) exponential stable equilibrium $\eta_e$, we say that the CBF $h$ satisfies the \textit{closed-loop (global) exponential minimum phase condition under $\pi$}.
\end{definition}

Note that with the restriction of $\phi = \Gamma \mu$, $\qcbf^{\pi}$ in Definition \ref{cond:minimum-phase-multi-input} depends solely on $\eta$, defining the closed-loop zero dynamics $\dot{\eta} = \qcbf^\pi (\eta)$ under $\pi$.
% \begin{equation}
%     \Psi_\phi^{-1}(\Gamma \mu, \eta)
% \end{equation}
% \begin{equation}
%     \pspi:=\mu(
% \end{equation}
Similarly, Theorem \ref{thm:min-phase-feedback} can be extended to the closed-loop dynamics $\qcbf^\pi$. By considering $\pspi(\eta) = \mu(x, \pi(x))$, we get the following corollary:

\begin{figure*}
\centering
\includegraphics[width=\textwidth]{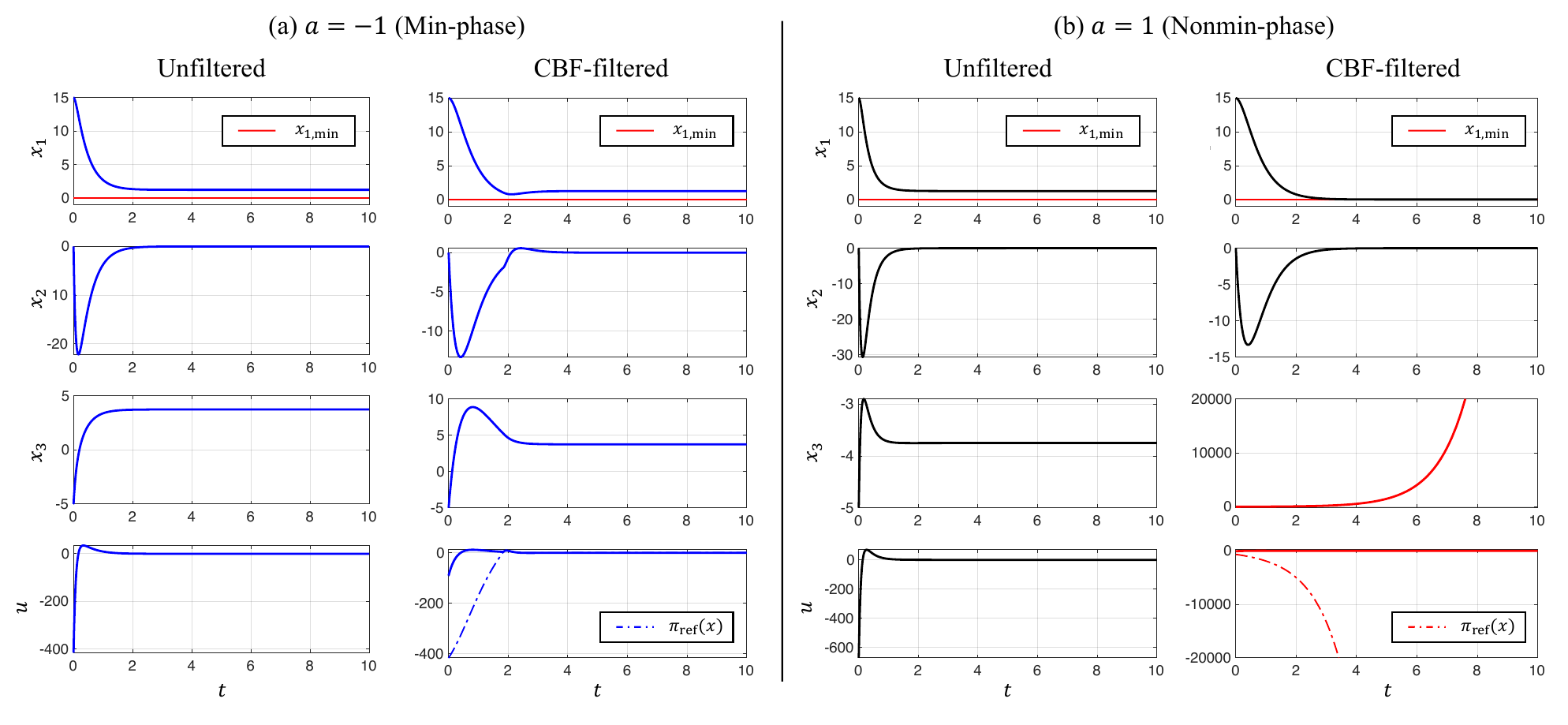}
\vspace{-2.25em}
\caption[\;]{Linear system example with single-input: state and input trajectories, under the LQR reference control input $\pi_{\textrm{ref}}(x)$ (Unfiltered), and the CBF-filtered control input $\pi_{\textrm{cbf}}(x)$ (CBF-filtered). The CBF filter is designed to ensure $x_1$ stays nonnegative, and $x_3$ is the internal state. Whereas (a) in the minimum phase condition ($a=-1$), the CBF-filtered system stays bounded, (b) in the nonminimum phase ($a=1$) condition, the CBF filter causes the internal state $x_3$ to explode\protect\footnotemark.}
\label{fig:linear-sys-ex}
\vspace{-1em}
\end{figure*}

\begin{corollary}
(Sufficient condition for boundedness of CBF-constrained dynamics---closed-loop multi-input case)
For the CBF $h$ which is Lipschitz continuous and the initial state $x_0 \in S_{\phi}$, let the closed-loop exponential minimum phase condition under $\pi$ (in Definition \ref{cond:minimum-phase-multi-input}) hold. Consider $V$ and $\alpha_{1, 2, 3, 4}$ defined in \eqref{eq:min-phase-lyap}. If $\alpha_4 > \left({\alpha_3 l_\phi \over 2}\right)^2$, and if $\gamma_{\min}$ is large enough, the system under the \barrierconstraint~can stay bounded. In particular, under $u=\pi(x)$, $\traj(t)$ converges to $x_e$, which is the solution to the equation,
\begin{equation}
    \Psi_{\phi} (x_e) = \big(\Gamma \mu(x_e, \pi(x_e)), \eta_e \big).
    \label{eq:equilibrium}
\end{equation}
\end{corollary}

\begin{proof}
This is a corollary of Theorem \ref{thm:min-phase-feedback} where \eqref{eq:equilibrium} defines the coordinate transformation from $x_e$ to $(\Gamma \pspi(\eta_e), \eta_e)$, which $(\trajphi(t), \trajeta(t))$ converges to as $t\rightarrow\infty$.
\end{proof}

In Section \ref{sec:multi-input-examples}, we will see through examples how the control input appearing in the internal dynamics can be utilized to shape the closed-loop minimum phase condition.

\vspace{-0.5em}
\section{Simulation Study: Single-Input Case}
\label{sec:simulations}

We now present two simulation examples---a linear and a nonlinear system---to verify the analysis we conducted in Section \ref{sec:minimum-phase-cbf} for single-input systems under the barrier constraint.

\vspace{-0.5em}
\subsection{Linear System}
\label{subsec:linear-system-result}

\subsubsection{Dynamics \& reference controller} Consider the problem of designing a stabilizing controller for a single-input linear system, described by $\dot{x} = Ax + Bu$, with
\begin{equation}
    A = \begin{bmatrix}
        0 & 1 & 0 \\
        0 & 0 & 0 \\
        3 & 1 & a
    \end{bmatrix}, \quad B = \begin{bmatrix}
        0 \\ 1 \\ 0
    \end{bmatrix},
\label{eq:linear-sys-ex}
\end{equation}
where $a$ is a model parameter. This system is controllable for any value of $a$, and any $x_e$ in the span of $[a, 0, -3]^\top$ is an equilibrium point of the undriven system. 

Consider a linear feedback controller $\pi_{\textrm{ref}}(x) = -K(x - x_e)$, whose feedback gain $K$ is derived as an LQR solution. Under this reference controller, $x$ converges to $x_{e}$ for any value of $a$. Figure \ref{fig:linear-sys-ex} (Unfiltered) shows such trajectories of the system when (a) $a=-1$ and (b) $a=1$, both successfully stabilizing to $x_{e}=[1.25, 0, 3.75]^\top$, $[1.25, 0, -3.75]^\top$, respectively. % In both cases, the system successfully stabilizes to the desired equilibria.

\subsubsection{Safety constraint, CBF, output dynamics} The safety constraint is to maintain $x_1$ to be nonnegative, i.e., $x_1(t) \ge 0\; \forall t\ge 0$. To achieve this for any reference controllers that might violate safety constraints---for instance, $\pi_{\textrm{ref}}$ with updated feedback gain $K$---we consider a CBF given as $h(x)\!=\!x_1$, whose relative degree $r$ is two, and the constraint coefficient vector $k = [6, 5]^\top$ is used, which results in $\gamma_1 = 2, \gamma_2 = 3$, and $\Gamma = [1\slash 6, 1\slash 3]^\top$ in \eqref{eq:def-Gamma}. We use a standard min-norm CBF filter which would filter $\pi_{\textrm{ref}}$ only when it violates the constraint, as
\vspace{-0.5em}
%That is, we apply
\begin{equation}
 \pi_{\textrm{cbf}}(x) = \max\left\{- (\gamma_1 \!+ \!\gamma_2) x_2 - \gamma_1 \gamma_2 x_1, \pi_{\textrm{ref}}(x) \right\}.
 \label{eq:single-input-ex-cbf-policy}
\vspace{-0.25em}
\end{equation}
% Note that the system is controllable and the safety filter minimally intervenes only when $\pi_{\textrm{ref}}(x)$ does not satisfy the barrier constraint. 

The CBF output derivative vector and the cascading constraint vector are given as
\begin{equation}
    \xi = \begin{bmatrix}
        x_1 \\ x_2
    \end{bmatrix}, \quad \phi = T \xi = \begin{bmatrix} 1 &  0 \\ 2 & 1 \end{bmatrix} \begin{bmatrix}
        x_1 \\ x_2
    \end{bmatrix} = \begin{bmatrix} x_1 \\ 2 x_1 + x_2\end{bmatrix}.
\end{equation}
The CBF output dynamics, represented in the cascading constraint form, is
\vspace{-0.5em}
\begin{equation}
        \dot{\phi} = \begin{bmatrix}
        -2 & 1 \\ 0 & -3
    \end{bmatrix} \phi + \begin{bmatrix}
        0 \\ 1 
    \end{bmatrix} \mu.
    \label{eq:linear-ex-cbf-output-dynamics}
\end{equation}

\subsubsection{CBF internal \& zero dynamics}
The internal coordinate can be easily found as $\eta = x_3$, and the CBF internal dynamics is expressed as \vspace{-0.5em} 
\begin{equation}
    \dot{\eta} = a \eta + \begin{bmatrix}
        1 & 1
    \end{bmatrix} \phi,
    \label{eq:single-input-example-zero} \vspace{-0.25em}
\end{equation}
With $\phi = \Gamma \mu$, the CBF zero dynamics becomes \vspace{-0.25em}
$$\dot{\eta} = a \eta + \frac{1}{2} \mu. \vspace{-0.25em}$$
Under the barrier constraint, $\mu \ge 0 $, the zero dynamics is stable only when $a < 0$, and is unstable when $a > 0$, which matches with Corollary \ref{cor:linear-system-case}.

%% FOOTNOTE FOR LINEAR EXAMPLE FIGURE
\footnotetext{In this example, the CBF safety filter intervenes in $\pi_{\textrm{ref}}(x)$ even when $x_1 \ge 0$ is satisfied under unfiltered $\pi_{\textrm{ref}}(x)$, since $\pi_{\textrm{ref}}(x)$ violates the barrier constraint \eqref{eq:barrier-constraint}, which constrain how quickly $x_1$ decrease to zero.}

\begin{figure*}
\centering
\includegraphics[width=\textwidth]{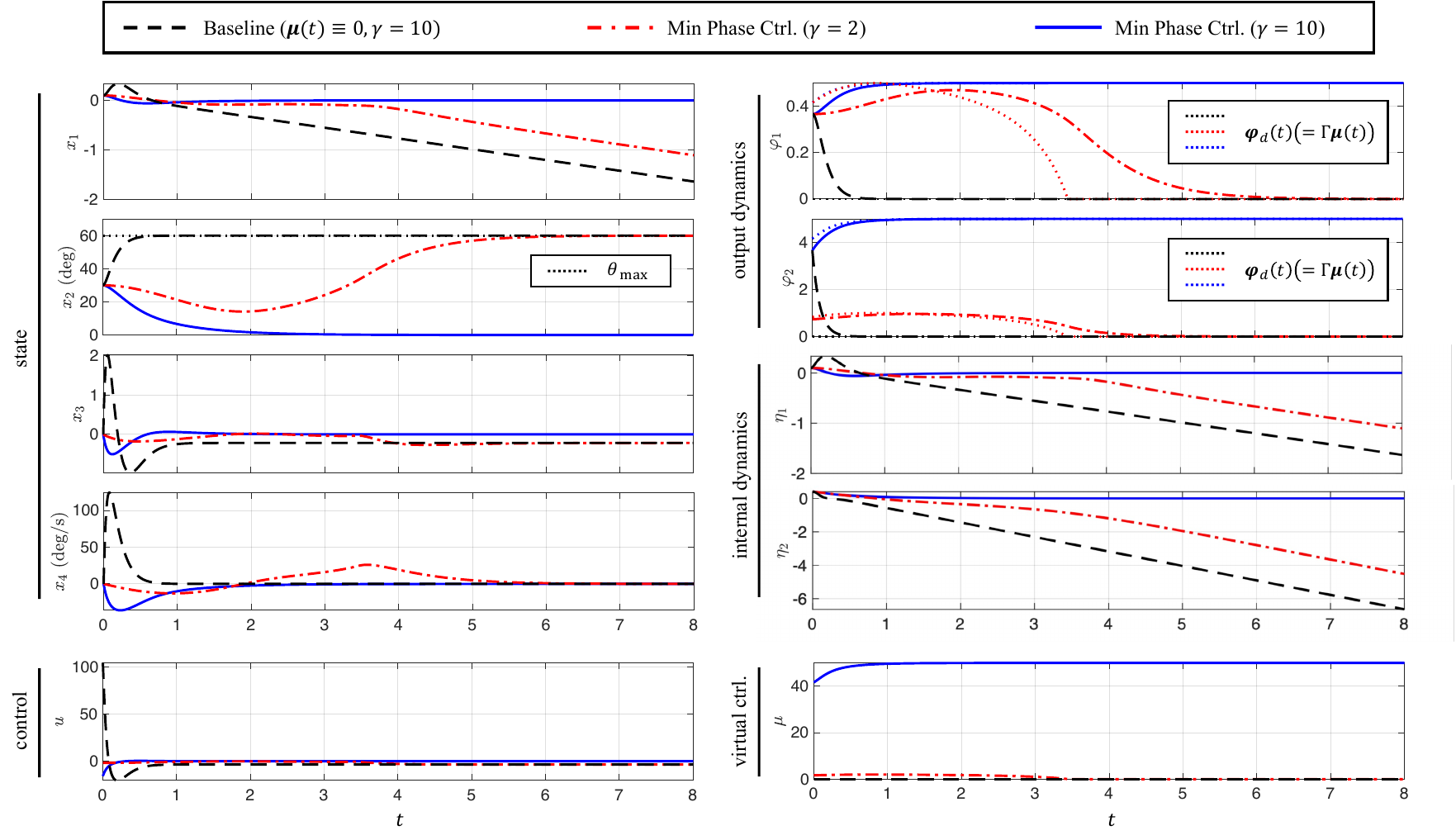}
\vspace{-1.75em}
\caption{Nonlinear system example with single input (cart-pole dynamics)---Baseline (black dashed): Barrier constraint imposed by naive safety filtering results in the internal state $\eta$ to drift. Min Phase Ctrl: Under the virtual control input feedback law that positively stabilizes the internal dynamics, with $\gamma=2$ (red dash-dotted), the internal state still drifts due to slow output dynamics response. When $\gamma=10$ (blue solid), the internal state is stabilized and the entire state stays bounded.}
\label{fig:carte-pole}
\vspace{-1.25em}
\end{figure*}

\subsubsection{Results}

Now, we test the CBF filter under the model parameters $a = -1$ and $1$. If $a < 0$, we can ensure the boundedness of $\eta$ according to Corollary \ref{cor:linear-system-case}. When $a = -1$ (Figure \ref{fig:linear-sys-ex} (a) CBF-Filtered), the internal state $x_3$ converges to $\eta_e = 3.75$ under $\mu_e = 7.5$ in Definition \ref{cond:minimum-phase-simple}. Although the CBF filter slightly overrides the reference controller, $x_1$ is still maintained nonnegative and $x$ converges to $x_{e}$, barely compromising $\pi_{\textrm{ref}}$. 

Next, we test the CBF filter under $a = 1$ (Figure \ref{fig:linear-sys-ex} (b) CBF-Filtered). While at first glance, the CBF-filter did its job in maintaining $x_1$ to be nonnegative, it not only failed to stabilize the system to $x_{e}$ (as safety is prioritized over stability), but the internal state $x_3$ of the CBF-constrained dynamics blows up. Whatever nonnegative signal $\trajmu(\cdot)$ is used, the internal dynamics will still diverge when the initial value of $\eta$ is positive according to Corollary \ref{cor:linear-system-case}. Thus, under this circumstance, the designer must not use the CBF safety filter.

\vspace{-0.5em}

\subsection{Nonlinear System Example}
\label{subsec:nonlinear-ex}

In this nonlinear example, we highlight that even under the minimum phase condition (Definition \ref{cond:minimum-phase-feedback}), 1) naive CBF-based safety filtering might result in an unbounded internal state; however, 2) this can be prevented by stabilizing the internal dynamics while satisfying the barrier constraint.

\subsubsection{Dynamics} We consider the nonlinear mechanics of a cart-pole system with an additional nonlinear drag term. The state consists of $x = [s, \theta, \dot{s}, \dtheta]^\top$, where $s$ is the cart position and $\theta$ is the pole angle with respect to the upright configuration. The control input $u$ is the horizontal force on the cart. The positional and angular accelerations are given as
\vspace{-0.5em}

\begin{subequations}
\label{eq:cart-pole}
    \begin{align}
        & \resizebox{0.87\hsize}{!}{$\displaystyle(1 \!+\!\sin^2\theta) \ddot{s}\! = - \dtheta^2 \sin{\theta}  + \cos\theta \sin\theta \!- \!\frac{b \dot{s}}{\cos\theta} \!(1 \!+ \!\sin^2\theta)\!+\!u,$} \label{eq:cart-pole1} \\
    & \resizebox{0.72\hsize}{!}{$\displaystyle(1 \!+ \!\sin^2\theta) \ddtheta\! = - \dtheta^2 \cos\theta \sin\theta  + 2 \sin\theta + \cos\theta \; u,$} \label{eq:cart-pole2}
    \end{align}
\end{subequations}

\noindent where $b=4$ is a drag coefficient. From \eqref{eq:cart-pole}, we can derive the state dynamics in the form of \eqref{eq:system}.

\subsubsection{Safety constraint, CBF, output dynamics}
The safety constraint we consider is to keep the pole angle within the range, $|\theta| \le \theta_{\max}$, where $\theta_{\max}$ is set to $60$\si{\degree}. To satisfy this constraint, we consider a CBF given as
\vspace{-0.25em}
\begin{equation}
    h(x) = \cos \theta - \cos \theta_{\max}.
    \label{eq:cart-pole-cbf}
    \vspace{-0.25em}
\end{equation}
This CBF has relative degree 2, and the CBF output derivative vector $\xi$ is given by 
\vspace{-0.25em}
\begin{equation}
\begin{aligned}
\xi_1 & = h(x) = \cos \theta - \cos \theta_{\max}, \\
\xi_2 & = \dot{\xi}_1 = -\dtheta \sin \theta.
\end{aligned}
\vspace{-0.25em}
\end{equation}

We set $\gamma = \gamma_1 = \gamma_2$, which results in $k$ in \eqref{eq:barrier-constraint} to be $[\gamma^2, 2\gamma]^\top$, and $\Gamma$ in \eqref{eq:def-Gamma} to be $[1\slash \gamma^2, \slash\gamma]^\top$. The virtual control input and the barrier constraint becomes
\vspace{-0.25em}
\begin{align}
\label{eq:cart-pole-mu}
    \mu & = \dot{\xi}_2 + 2\gamma \xi_2 + \gamma^2 \xi_1 \\
    & = \resizebox{0.9\hsize}{!}{$\displaystyle-\ddtheta \sin\theta - \dtheta^2 \cos\theta - 2 \gamma \dtheta \sin\theta + \gamma^2 (\cos\theta - \cos\theta_{\max}) \ge 0,$} \nonumber \vspace{-0.25em}
\end{align}
where $\ddtheta$ is determined by $u$ from \eqref{eq:cart-pole2}.

The cascading constraint vector $\phi$ is given by $\phi_1 = \xi_1$, $\phi_2 = \xi_2 + \gamma \xi_1$. When $\phi$ is constrained on the equilibrium manifold, $\Phi_e$ \eqref{eq:equilibria-manifold},  $\phi_2 = \gamma \phi_1$, and $\xi_2 = -\dtheta \sin \theta = 0$.

\subsubsection{CBF internal \& zero dynamics}

Based on Condition \ref{cond:nonactuated_internal_state}, we can find the internal state $\eta$ as
\vspace{-0.25em}
\begin{equation}
    \eta_1 = s, \quad
    \eta_2 = \dot{s} \cos\theta - \dtheta + b s, \vspace{-0.25em}
\label{eq:cart-pole-internal-state}
\end{equation}
and the internal dynamics becomes
\begin{equation}
\begin{aligned}
\dot{\eta}_1 & = \frac{1}{\cos\theta}\left(- b \eta_1 + \eta_2 + \dtheta \right), \\
\dot{\eta}_2 & = -\sin\theta \left(1 + \dot{s} \dtheta\right).
\end{aligned}
\end{equation}

If $\theta \neq 0$, on the CBF zero dynamics manifold ($\phi \in \Phi_e$), $\dtheta = 0$, since $\dtheta \sin \theta = 0$. Thus, the zero dynamics becomes
\begin{equation}
\label{eq:cart-pole-zero-dynamics}
\begin{aligned}
\dot{\eta}_1 & = \frac{1}{\cos\theta}\left(- b \eta_1 + \eta_2 \right), \\
\dot{\eta}_2 & = -\sin\theta.
\end{aligned}
\end{equation}

\subsubsection{Design of positive-feedback-stabilizable minimum phase condition}
\label{subsec:cart-pole-min-phase-design}

We design $\pspi(\eta)$ in Theorem \ref{thm:min-phase-feedback} indirectly by designing a desired $\theta$ that would result in a minimum phase zero dynamics. In \eqref{eq:cart-pole-zero-dynamics}, if $\theta$ becomes
\begin{equation}
    \theta_d(\eta) := \arcsin \left( \max \{ \min\{\eta_2, \sin\theta_{\max}\}, -\sin\theta_{\max} \} \right),
\label{eq:theta_d}
\end{equation}
we have the zero dynamics being
\begin{equation}
\label{eq:cart-pole-zero-dynamics-feedback}
\begin{aligned}
\dot{\eta}_1 & = \frac{1}{\cos\theta}\left(- b \eta_1 + \eta_2 \right), \\
\dot{\eta}_2 & = - \max \{ \min\{\eta_2, \sin\theta_{\max}\}, -\sin\theta_{\max} \}.
\end{aligned}
\end{equation}
When $|\eta_2|$ is not saturated to $\sin\theta_{\max}$, the Jacobian of \eqref{eq:cart-pole-zero-dynamics-feedback} becomes
\begin{equation*}
    \begin{bmatrix}
        -b \slash \cos\theta & 1 \slash \cos\theta \\
        0 & -1
    \end{bmatrix},
\end{equation*}
and since $\cos \theta \ge \cos \theta_{\max} > 0$, the zero dynamics becomes exponentially stable. The corresponding virtual control law for $\mu$ is
\begin{equation}
\label{eq:cart-pole-min-phase-design}
    \pspi(\eta) = \gamma^2 \left(\cos\theta_d(\eta) - \cos\theta_{\max} \right),
\end{equation}
which is nonnegative since $|\theta_d(\eta)| \le \theta_{\max}$ from \eqref{eq:theta_d}. We can compute $u$ that results in $\mu \equiv \pspi(\eta)$ from \eqref{eq:cart-pole2} and \eqref{eq:cart-pole-mu}.

\subsubsection{Results}
We test three controllers in simulation from the initial state $x_0 = [0.1, \pi\slash6, 0, 0]^\top$, all of which are constrained by the barrier constraint \eqref{eq:cart-pole-mu}. 
When a naive CBF-based safety filter is implemented, if the reference $u$ does not satisfy the barrier constraint, the barrier constraint will continue to saturate ($\mu =0$). We first test this case as our baseline. As shown in Figure \ref{fig:carte-pole} (Baseline), although $x_2$ is maintained below $\theta_{\max}$, the internal state $\eta$ and the cart position continues to drift.

We next impose the minimum phase condition in Definition \ref{cond:minimum-phase-feedback} through the feedback control designed in \eqref{eq:cart-pole-min-phase-design}. We test this controller under two different values of the CBF convergence rate, $\gamma=2$, and $10$. As shown in Figure \ref{fig:carte-pole}, we observe that when $\gamma = 2$ (red), the drift of the internal state $\eta$ is still not prevented. In contrast, when $\gamma = 10$ (blue), the internal dynamics is successfully stabilized, and all state variables stay bounded. Recall that in Theorem \ref{thm:min-phase-feedback}, the boundedness is only guaranteed when $\gamma$ is large enough. Our simulation result confirms the necessity of this condition.

\section{Simulation Study: Multi-Input Case}
\label{sec:multi-input-examples}
In this section, by introducing additional control inputs to the previous single-input linear and nonlinear examples, we will see how we can shape the multi-input closed-loop minimum phase condition by designing $\pi$ in Definition \ref{cond:minimum-phase-multi-input}.

\vspace{-0.5em}

\subsection{Linear System Example}
\label{subsec:multi-input-linear}

We add an additional control input to the single-input system in \eqref{eq:linear-sys-ex}, making $m=2$, given as
\begin{equation}
    A = \begin{bmatrix}
        0 & 1 & 0 \\
        0 & 0 & 0 \\
        3 & 1 & a
    \end{bmatrix}, \quad B = \begin{bmatrix}
        0 & 0 \\ 1 & 0\\ 0 & 1
    \end{bmatrix}.
\label{eq:linear-sys-ex-mi}
\end{equation}
With the same CBF $h(x) = x_1$ considered, the CBF output dynamics remains the same as \eqref{eq:linear-ex-cbf-output-dynamics}. We use the same values of $\gamma_1 = 2$, $\gamma_2 = 3$ as before, resulting in the CBF virtual input $\mu = u_1 + 6 x_1 + 5 x_2$. We consider $a=1$ in this multi-input case, which led to the unbounded internal state in the previous single-input case. 

Considering the internal coordinate $\eta = x_3$, we get the internal dynamics $\dot{\eta} = \eta + \begin{bmatrix}
1 & 1
\end{bmatrix} \phi + u_2$,
% \vspace{-0.25em}
% \begin{align}
% \dot{\eta} = \eta + \begin{bmatrix}
% 1 & 1
% \end{bmatrix} \phi + u_2,
% \vspace{-0.25em}
% \end{align}
and the zero dynamics $\dot{\eta} = \eta + \frac{1}{2} \mu + u_2.$
% \vspace{-0.25em}
% \begin{align}
% \dot{\eta} = \eta + \frac{1}{2} \mu + u_2.
% \vspace{-0.25em}
% \end{align}
Now we can utilize $u_2$ to impose a stability condition on the internal dynamics. We design an equality condition given as
\vspace{-0.25em}
\begin{equation}
    \frac{1}{2} u_1 + 3 x_1 + \frac{5}{2}x_2 + u_2 = -2 x_3,
    \label{eq:linear-mimo-ex-internal}
    \vspace{-0.25em}
\end{equation}

\begin{figure}[t]
\centering
\vspace{-1em}
\includegraphics[width=\columnwidth]{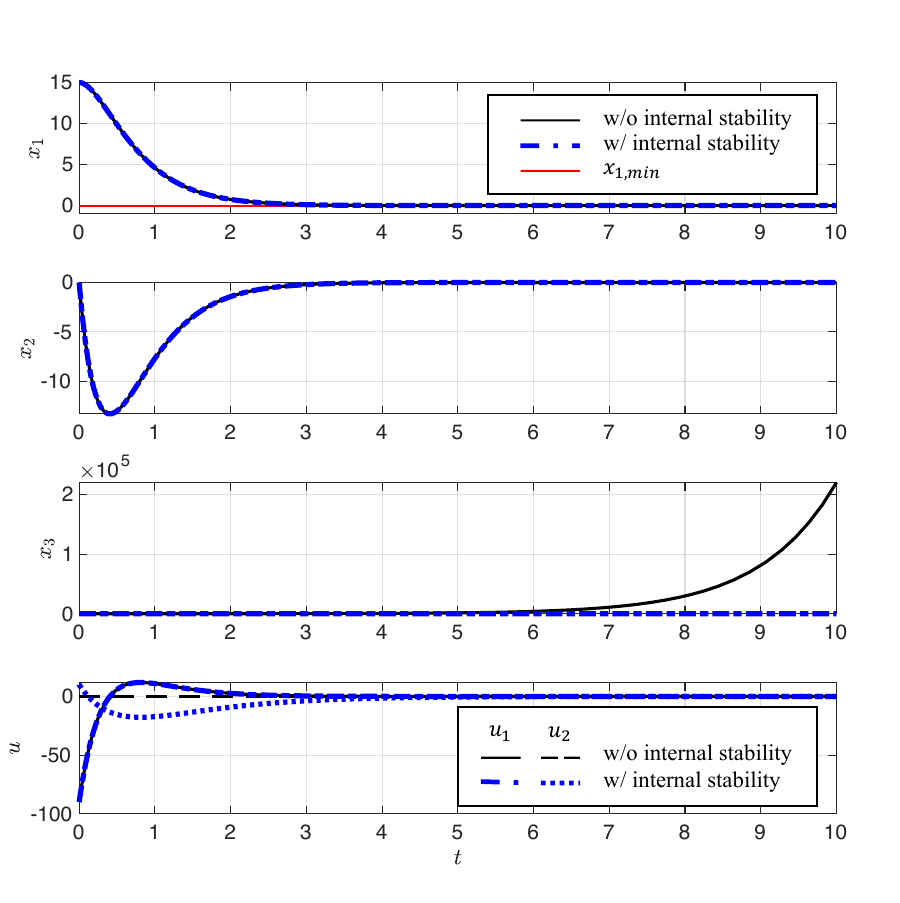}
\vspace{-2.25em}
\caption{Linear system example with multi-input. While the naive CBF safety filter results in unstable internal dynamics (black), by utilizing the additional input $u_2$, we can design a CBF safety filter that incorporates the additional condition \eqref{eq:linear-mimo-ex-internal}, designed for the stability of the internal dynamics (blue).}
\label{fig:linear-mimo}
\vspace{-1em}
\end{figure}

\noindent which is equivalent to $\frac{1}{2} \mu + u_2 = -2 \eta$. Thus, the equality condition results in the zero dynamics of $\dot{\eta} = -\eta$, which becomes the exponential minimum phase condition.

\begin{figure*}
\centering
\includegraphics[width=\textwidth]{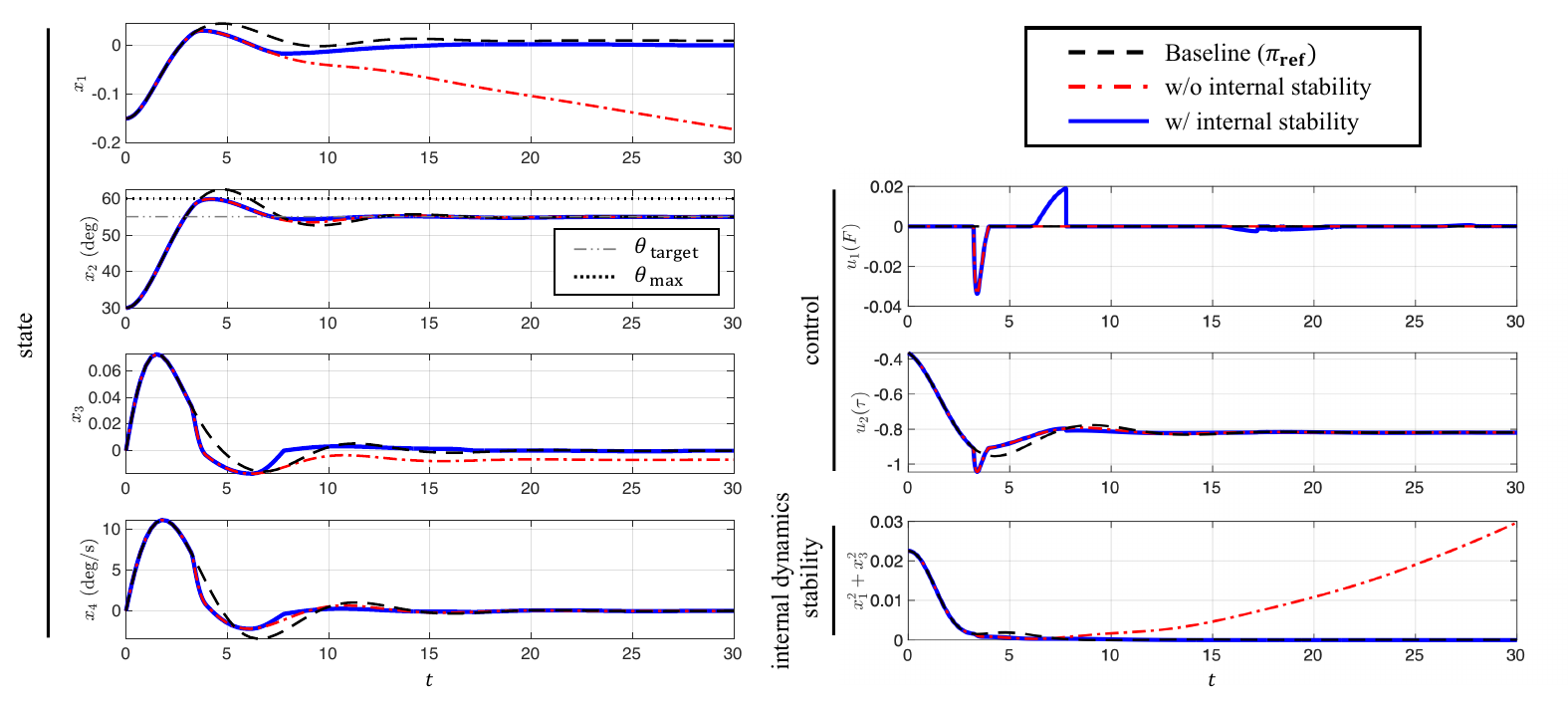}
\vspace{-1.5em}
\caption{Nonlinear system example with multi-input (cart-pole dynamics with force and torque)---(Black dashed) The reference feedback linearization controller stabilizes pole angle $x_2$ to a target value $55^\circ$, while violating the constraint $x_2 \le \theta_{\max} = 60^\circ$. (Red dash-dotted) The plain CBF safety filter prevents the safety violation, but as a result, $\eta_1=x_1$ drifts, causing instability in internal dynamics. (Blue solid) The additional internal stability condition incorporated into the CBF safety filter stabilizes the internal dynamics and prevents the cart from drifting, while satisfying the safety constraint.}
\label{fig:carte-pole-mimo}
\vspace{-1em}
\end{figure*}

The barrier constraint inequality and the internal stability condition in \eqref{eq:linear-mimo-ex-internal} can be imposed together into the CBF safety filter, which becomes a quadratic program with one linear inequality and one linear equality constraint. The trajectory under the filtered control signal, with the reference controller we designed in Section \ref{subsec:linear-system-result}, is shown in Figure \ref{fig:linear-mimo}.

\vspace{-0.5em}
\subsection{Nonlinear System Example}
\label{subsec:multi-input-nonlinear}

We modify the cart-pole example in Section \ref{subsec:nonlinear-ex} by introducing a torque between the cart and the pole. We remove the drag term, resulting in a nonminimum-phase condition in the previous single-input case. The dynamics becomes
\begin{subequations}
\label{eq:cart-pole-mimo}
    \begin{align}
        &\resizebox{0.82\hsize}{!}{$\displaystyle(1 \!+\!\sin^2\theta) \ddot{s}\! = - \dtheta^2 \sin\theta  + \cos\theta \sin\theta \! + u_1 + \cos\theta u_2,$} \label{eq:cart-pole-mimo1}\\
    &\resizebox{0.87\hsize}{!}{$\displaystyle(1 \!+ \!\sin^2\theta) \ddtheta\! = - \dtheta^2 \cos\theta \sin\theta  + 2 \sin\theta + \cos\theta u_1 + 2 u_2,$} \label{eq:cart-pole-mimo2}
    \end{align}
\end{subequations}
where $u_1$ is the horizontal force, and $u_2$ is the torque input. Since there is no drag, the linear momentum can only be altered through $u_1$, determining the evolution of $s$ and $\dot{s}$.

We consider a reference controller $\pi_{\text{ref}}$ that stabilizes the pole angle $\theta$ to a target value $\theta_{\scriptstyle\textrm{target}}=55^\circ$, while enforcing the safety constraint $|\theta| \!\le\! \theta_{\max}\!=\!60^\circ$. % , where $\theta_d = 55^\circ$, $\theta_{\max} = 60^\circ$. 
The reference controller uses only the torque $u_2$ and is designed using an IO linearization control law with output $y = \theta - \theta_{\scriptstyle\textrm{target}}$ (Fig. \ref{fig:carte-pole-mimo} black). Although this controller drives $\theta$ to $\theta_{\scriptstyle\textrm{target}}$, it violates the safety constraint during the transient phase. To prevent this, we apply a min-norm CBF safety filter, using the CBF defined in \eqref{eq:cart-pole-cbf} (Fig. \ref{fig:carte-pole-mimo} red). The filter overrides the reference input by intervening on $u_1$, which leads to the accumulation of linear momentum and instability in the CBF internal dynamics.

According to Theorem \ref{thm:multi-input-internal}, due to the additional input, the internal dynamics depends on $u$. To simplify the analysis, we instead consider an alternative internal state defined as $\eta = [s\; \dot{s}]^\top$. The resulting internal dynamics is

\;\vspace{-2.5em}
\begin{align}
\dot{\eta}_1 & = \eta_2, \label{eq:cart-pole-mimo-internal} \\
\dot{\eta}_2 & = \frac{- \dtheta^2 \sin\!\theta  + \cos\!\theta \sin\!\theta}{1 + \sin^2\!\theta} + \frac{1}{1 + \sin^2\!\theta}\begin{bmatrix}
    1 & \cos\!\theta
\end{bmatrix}\begin{bmatrix}
    u_1 \\ u_2
\end{bmatrix}. \nonumber
\end{align}

Using this internal state, we define a Lyapunov function candidate $V(\eta) = \eta_1^2 + \eta_2^2$ to design a safety filter that ensures the boundedness of $\eta$. We adopt the CLF-CBF-QP formulation from \eqref{eq:clf-cbf-qp} by replacing $W(x)$ with $V(\eta)$. The Lie derivatives of $V$ can be evaluated based on \eqref{eq:cart-pole-mimo-internal}. This results in a safety filter that incorporates both the CBF-based barrier constraint and a CLF-based stability constraint on the internal state.

As shown in Fig.~\ref{fig:carte-pole-mimo} (blue), the CLF-CBF-QP safety filter achieves three goals: 1) it successfully maintains boundedness of all state variables, 2) while still achieving the original control objective of driving $\theta$ to $\theta_{\scriptstyle\textrm{target}}$, 3) without exceeding the safety threshold $\theta_{\max}$. The plot of $u_1(t)$ shows that the additional CLF condition ensures the cancellation of the accumulated linear momentum resulting from the barrier constraint's intervention, thereby preventing the cart from drifting. % Finally, we observe that $\phi(t)$ converges to $\Gamma \mu(t)$ at the end of the simulation, verifying that the system state converges to the equilibrium $x_e$ in \eqref{eq:equilibrium}.

\vspace{-0.5em}
\section{Conclusion \& Future Work}
\label{sec:conclusion}

In this work, we investigated the internal dynamics of systems governed by CBF-constrained control laws and proposed a new set of CBF minimum phase conditions. When a given CBF satisfies one of these conditions, along with certain additional requirements, the resulting CBF-constrained dynamics can be guaranteed to remain bounded. Central to our formulation is the new notion of CBF zero dynamics, in which the virtual control input must remain nonnegative, reflecting the enforcement of the barrier constraint. From this perspective, the boundedness of the state can be ensured by the positive stabilizability of the CBF zero dynamics with a small number of additional conditions.

While our analysis and the proposed CBF minimum phase conditions establish a step toward certifying the safe use of CBF-based safety filters without introducing undesirable behaviors, several important directions remain for future work.
First, in the single-input case, while our minimum phase conditions serve as sufficient conditions for boundedness, it remains an open question whether these conditions can be further relaxed. Note that in the multi-input case, additional input variables that appear in the internal dynamics can be used to actively stabilize the internal state. Second, although our conditions are framed in terms of positive stabilizability, a constructive approach to verifying this property must still be developed. One possible direction is to verify the existence of a Lyapunov function (e.g., as in \eqref{eq:min-phase-lyap}) certifying positive stabilizability. Finally, our current analysis assumes unbounded control inputs, consistent with the original formulation of exponential CBFs \cite{nguyen2016exponential}. However, in many realistic systems, control bounds are themselves safety constraints, and such bounds can limit the ability to positively stabilize the internal dynamics. We believe that constructive techniques—such as sum-of-squares (SOS) programming \cite{dai2023convex, zhao2023convex}, which can incorporate nonnegativity and control bounds explicitly—hold promise for deriving practical, certifiable conditions under which CBF-based safety filters can be deployed without risking unbounded system behavior.

% \section*{Appendix}
\vspace{-0.5em}
\appendix

% \setcounter{subsection}{0}
% \vspace{-0.5em}
\subsection{Lemmas}
\label{appen:lemmas}

Since $A_{\gamma}$ is Hurwitz, there exists the Lyapunov matrix $P$ that is positive definite, which is the solution of a Lyapunov equation,
\begin{equation}
    A_{\gamma}^\top P + P A_{\gamma} = -I.
\label{eq:lyap}
\end{equation}

\begin{lemma}[Boundedness of $\trajphidiff$] 
\label{lemma:boundedness-phidiff}
For a Lipschitz continuous virtual control input signal $\trajmu(\cdot)$ such that $\trajmu(t) \ge 0$ for all $t \ge 0$, $\trajphidiff(\cdot)$ stays bounded.
\end{lemma}
\begin{proof}
We set $W(\phidiff):=\phidiff^\top P \phidiff$, where $P$ is given by \eqref{eq:lyap}.
From \eqref{eq:error-dynamics}, we have
\begin{align*}
    \dot{W}(\phidiff, t) = -\phidiff^\top \phidiff - 2 \phidiff^\top P \Gamma \dot{\trajmu}(t).
\end{align*}
From the Lipschitz continuity of $\trajmu$, there exists a constant $M > 0$ such that $|\dot{\trajmu}(t)| \le M$ for all $t\ge 0$. Thus, we have
\begin{equation}
    \dot{W}(\phidiff, t) \le -\phidiff^\top \phidiff + C |\phidiff|,
\end{equation}
for some constant $C>0$. Thus, we can prove that $\phidiff$ is confined to a ball of radius $\sqrt{\frac{\lambda_{\max}(P)}{\lambda_{\min}(P)}}C$.
\end{proof}

\begin{lemma}[Bound of $\|P\|$] 
\label{lem:P} For $P$ defined as the solution to \eqref{eq:lyap}, there exists some constant $C$ such that $\|P\| \le \frac{C}{\gamma_{\min}}$.
\end{lemma}

\begin{proof}
    The solution to the Lyapunov equation \eqref{eq:lyap} can be expressed as $P = \int_{0}^{\infty} e^{A_{\gamma}^\top t} e^{A_{\gamma} t} dt$.
% \begin{equation*}
%     \|P\| = \left\| \int_{0}^{\infty} e^{A_{\gamma} t} e^{A_{\gamma}^\top t} dt. \right\| \le \int_{0}^{\infty} \|e^{A_{\gamma} t}\|^2 dt.
% \end{equation*}
There exists some constant $C'$ such that $\| e^{A_{\gamma} t}\| \le C' e^{-\gamma_{\min} t}$, since $-\gamma_{\min}$ is the maximum eigenvalue of $A_{\gamma}$. Thus,
\begin{equation*}
    \|P\| \le \int_{0}^{\infty} \|e^{A_{\gamma} t}\|^2 dt \le C'^2 \int_{0}^{\infty} e^{-2\gamma_{\min}t}dt = \frac{C'^2}{2\gamma_{\min}}.
\end{equation*}
\vspace{-0.5em}
\end{proof}

\begin{lemma}[Bound of $\|\Gamma\|$]
\label{lem:Gamma} 
For $\Gamma$ defined in \eqref{eq:def-Gamma}, when $\gamma_{\min} \ge \sqrt{2}$, \vspace{-0.5em}
\begin{equation}
    \| \Gamma\| \le \frac{\sqrt{2}}{\gamma_{\min}}.
\end{equation}
\end{lemma}
\begin{proof}
Note that $\gamma_r$ is set to $\gamma_{\min}$ without loss of generality. $(\gamma_k \cdots \gamma_r)^{-1} \le \gamma_r^{r-k+1}$, thus,
\vspace{-0.5em}
\begin{align*}
\| \Gamma\| \le \sqrt{\frac{\gamma_r^{2r}-1}{\gamma_r^{2r} (\gamma_r^{2}-1)}} \le \sqrt{\frac{1}{\gamma_r^2-1}}.
\vspace{-0.5em}
\end{align*}
It is easy to see that
$\sqrt{\frac{1}{\gamma_r^2-1}}$ is less than $\frac{\sqrt{2}}{\gamma_r}$ when $\gamma_r \!\ge\! \sqrt{2}$.
% \begin{align*}
%     \| \Gamma\| &= \| - A_{\gamma}^{-1} B \| \le \|-A_{\gamma}^{-1}\| \cdot \|B\| \le \frac{1}{\gamma_{\min}} \cdot 1 = \frac{1}{\gamma_{\min}},
% \end{align*}
% by checking that the minimum eigenvalue of $-A_{\gamma}$ is $\gamma_{\min}$.
\end{proof}

\begin{lemma}[Bound of $\|P \Gamma\|$]
\label{lem:PLambda} For $P$ defined as the solution to \eqref{eq:lyap}, and $\Gamma$ defined in \eqref{eq:def-Gamma}, when $\gamma_{\min} \ge \sqrt{2}$,
\begin{equation}
    \|P \Gamma\| \le \frac{C'}{\gamma_{\min}^2},
\end{equation}
for some constant $C'$ that does not depend on $\gamma_i$s.
\end{lemma}
\begin{proof}
\; \vspace{-2em}
\begin{align*}
    \|P \Gamma\| & \le \| P \| \cdot \|\Gamma\| \le \frac{C}{\gamma_{\min}} \cdot \frac{\sqrt{2}}{\gamma_{\min}} = \frac{C'}{\gamma_{\min}^2},
\end{align*}
where the second inequality is from Lemmas \ref{lem:P} and \ref{lem:Gamma}.
\end{proof}

\subsection{Proof of Theorem \ref{thm:min-phase-simple-global}}
\label{appendix:simple-global-proof}

The proof resembles that of the feedback linearization bounded tracking theorem \cite[Theorem 9.14]{sastry2013nonlinear}. 

From \eqref{eq:error-dynamics}, we have $\dot{\trajphidiff}(t) = A_{\gamma} \trajphidiff(t)$ under $\trajmu(t)\equiv \mu_e$. Since $A_{\gamma}$ is Hurwitz, $\trajphidiff(t)$ exponentially converges to zero, i.e. $|\trajphidiff(t)| \le C e^{-\gamma_{\min}t}$, for some constant $C$ (and thus $\trajphi(t)$ converges to $\Gamma \mu_e$).

Next, we prove $\trajeta(t) \rightarrow \eta_e$. Without loss of generality, we set $\eta_e=0$ in the proof. Due to Condition \ref{cond:minimum-phase-simple}, based on converse Lyapunov theorem \cite[Theorem 5.17]{sastry2013nonlinear}, there exists a Lyapunov function $V$ satisfying
\begin{equation}
\label{eq:exp_lyap1}
\begin{aligned}
    \alpha_1 |\eta|^2 \le V(\eta) \le \alpha_2 |\eta|^2, \quad \left|\frac{\partial V}{\partial \eta}\right| \le \alpha_3 |\eta|, \\
    \frac{\partial V}{\partial \eta} \cdot \qcbf (\eta, \Gamma \mu_e) \le - \alpha_4 |\eta|^2,
\end{aligned}
\end{equation}
with $\alpha_{1, 2, 3, 4} > 0$. The time derivative of $V$ becomes
\begin{align*}
    \dot{V} &= \frac{\partial V}{\partial \eta} \cdot \qcbf (\eta, \phi) \\
    & = \frac{\partial V}{\partial \eta} \cdot \qcbf (\eta, \Gamma \mu_e) + \frac{\partial V}{\partial \eta} \cdot \left\{ \qcbf (\eta, \phi) - \qcbf (\eta, \Gamma \mu_e)\right\} \\
    & \le -\alpha_4 |\eta|^2 + \alpha_3 |\eta| l_\phi |\Delta\phi| \\
    & \le -\alpha_4 |\eta|^2 + \alpha_3 l_\phi  C e^{-\gamma_{\min} t} |\eta|,
\end{align*}
where $l_\phi$ is the Lipschitz constant of $\qcbf$ with respect to $\phi$. Thus, it follows that
\begin{align*}
    \dot{V}(\eta) \le -\alpha_4 \!\left(|\eta| - \frac{\alpha_3 l_\phi C e^{-\gamma_{\min} t}}{2 \alpha_4} \right)^2\! + \frac{\alpha_3^2 l_\phi^2 C^2 e^{-2\gamma_{\min} t}}{4 \alpha_4}.
\end{align*}
Thus, $\dot{V}$ is negative for all $\eta$ satisfying $|\eta| > \frac{\alpha_3 l_\phi C e^{-\gamma_{\min} t}}{\alpha_4}$. From \eqref{eq:exp_lyap1}, we can verify that $\eta$ is confined to a ball of radius $\sqrt{\frac{\alpha_2}{\alpha_1}} \frac{\alpha_3 l_\phi C e^{-\gamma_{\min} t}}{\alpha_4}$, which exponentially reduces to 0 as $t \rightarrow \infty$. Thus, $\trajeta(t)$ converges to 0.

\subsection{Proof of Theorem \ref{thm:min-phase-feedback}}
\label{appendix:feedback-theorem-proof}

\begin{proof}
Recall that $\trajphid(t) := \Gamma \trajmu(t) = \Gamma \pspi(\trajeta(t))$, and recall the error dynamics in \eqref{eq:error-dynamics}. Then,
\begin{align}
\dot{V} &= \frac{\partial V}{\partial \eta} \cdot \qcbf (\eta, \phi) \label{eq:Vdot}\\
& = \frac{\partial V}{\partial \eta} \qcbf (\eta, \Gamma \pspi(\eta)) + \frac{\partial V}{\partial \eta} \left\{\qcbf (\eta, \phi) - \qcbf(\eta, \Gamma \pspi(\eta)) \right\} \nonumber \\
& \le - \alpha_4 |\eta|^2 + \alpha_3 l_\phi |\phidiff| |\eta| ,  \quad \text{(from Lipschitz continuity)}\nonumber 
\end{align}

Next, we set $U(\phidiff):=\phidiff^\top P \phidiff$, where $P$ is given by \eqref{eq:lyap}. We will prove the theorem by considering a Lyapunov function candidate
\begin{equation}
    W(x):= U(\phidiff(x)) + V(\eta(x)),
    \label{eq:lyapunov-total}
\end{equation}
and by evaluating its time derivative under $\mu=\pspi(\eta)$.

From
\vspace{-0.5em}
\begin{equation}
    \dot{\mu} = \dot{\pspi}(\eta) = \frac{\partial\pspi}{\partial\eta}\dot{\eta} = \frac{\partial\pspi}{\partial\eta} \qcbf(\eta, \phi),
\end{equation}
and \eqref{eq:error-dynamics}, we have $\dot{\phidiff} = A_{\gamma} \phidiff - \Gamma \frac{\partial\pspi}{\partial\eta} \qcbf(\eta, \phi)$.
Thus,
\begin{align*}
    \dot{U} & = \phidiff^\top (A_{\gamma}^\top P + P A_{\gamma}) \phidiff  - 2 \phidiff^\top P \Gamma \frac{\partial\pspi}{\partial\eta} \qcbf(\eta, \phi) \\
    & = -\phidiff^\top \phidiff  - 2 \phidiff^\top P \Gamma \frac{\partial\pspi}{\partial\eta} \qcbf(\eta, \phi)
\end{align*}
Here, we note that 1) $\left|\frac{\partial\pspi}{\partial\eta}\right| \le l_{\kappa}$, where $l_{\kappa}$ is the Lipschitz constant of $\pspi$ w.r.t $\eta$, 2) $\qcbf(0, \pspi(0))=0$ (as $\eta_e = 0$), and 3) 
\vspace{-1em}
\begin{align}
    |\qcbf(\eta, \phi)| = & \left|\qcbf(\eta, \phi) - \qcbf(0, \pspi(0))\right| \nonumber \\
    =  & \big|\big(\qcbf(\eta, \phi) - \qcbf\left(\eta, \Gamma \pspi(\eta)\right) \big) \nonumber \\
    & + \big(\qcbf\left(\eta, \Gamma \pspi(\eta)\right) - \qcbf\left(\eta, \Gamma \pspi(0)\right)\big) \nonumber \\
    & + \big(\qcbf\left(\eta, \Gamma \pspi(0)\right) - \qcbf\left(0, \Gamma \pspi(0)\right) \big)\big| \nonumber \\
    \le & l_\phi |\phidiff| + l_\phi |\Gamma| l_{\kappa} |\eta| + l_\eta |\eta|,
\end{align}
where $l_\eta$ is the Lipschitz constant of $\qcbf$ with respect to $\eta$. Thus,
\vspace{-0.5em}

{\small
\begin{align}
    \dot{U} & \le -|\phidiff|^2 \!+\!2 |\phidiff| \; \|P\Gamma\|\; l_{\kappa} \left(l_\phi |\phidiff|\!+\!l_\phi l_{\kappa} |\Gamma| |\eta|\!+\!l_\eta |\eta|\right) \nonumber \\
    & \le  -\left(1 - \frac{C_1}{\gamma_{\min}^2}\right) |\phidiff|^2 + \frac{2(C_2 + C_3 \gamma_{\min})}{\gamma_{\min}^2} |\phidiff||\eta|, \label{eq:Udot}
\end{align}
}

\noindent for some positive constants $C_1$ and $C_2$ that does not depend on $\gamma_i$'s. The second inequality results from Lemmas \ref{lem:Gamma} and \ref{lem:PLambda} in Appendix \ref{appen:lemmas}.

From \eqref{eq:Vdot} and \eqref{eq:Udot}, we get 
{\small 
\begin{align}
\label{eq:proof-uvdot}
    \dot{U} + \dot{V} \le & -\begin{bmatrix}
        |\eta| & |\phidiff|
    \end{bmatrix}^T \\
    & \resizebox{0.6\hsize}{!}{$\displaystyle\begin{bmatrix}
       \alpha_4 &  -\Big(\frac{C_2}{\gamma_{\min}^2}\!+\!\frac{C_3}{\gamma_{\min}}\!+\! C_4\!\Big) \\ -\Big(\frac{C_2}{\gamma_{\min}^2}\!+\!\frac{C_3}{\gamma_{\min}}\!+\!C_4\!\Big) & 1 - \frac{C_1}{\gamma_{\min}^2}
    \end{bmatrix}$} \!\begin{bmatrix}
        |\eta| \\ |\phidiff| \nonumber
    \end{bmatrix}
\end{align}
}
where $C_4:= {\alpha_3 l_\phi \over 2}$.

Finally, we can prove the statement by showing that the square matrix in the right hand side of \eqref{eq:proof-uvdot} can be made positive definite. The necessary and sufficient condition for this is that $\alpha_4>0$, $1 - \frac{C_1}{\gamma_{\min}^2} >0$, and the determinant of the matrix being positive. The first condition is given by the statement, and the second condition is satisfied by taking $\gamma_{\min} > \sqrt{C_1}$. Finally, for the determinant to be positive, we need \vspace{-0.5em}
\begin{equation}
\label{eq:proof-gamma-poly}
    \resizebox{0.87\hsize}{!}{$\displaystyle\alpha_4 \gamma_{\min}^4 - C_1 \alpha_4 \gamma_{\min}^2 - (C_2 + C_3 \gamma_{\min} + C_4 \gamma_{\min}^2)^2 > 0,$}
\end{equation}
which is a quartic polynomial in $\gamma_{\min}$ whose fourth-order coefficient is given as $\alpha_4 - C_4^2$. Since $\alpha_4 > C_4^2$, by the condition given in the statement, we can satisfy \eqref{eq:proof-gamma-poly} by taking sufficiently large $\gamma_{\min}$.
\end{proof}

\vspace{-0.5em}
\subsection{Proof of Theorem \ref{thm:multi-input-internal}}
\label{appen:theorem-multi}
We define the distribution $G:=\textrm{span}\{g_1, \cdots, g_m\}$. When $g_1, \cdots, g_m$ are linearly independent at $x$, the rank of $G$ is $m$. We also define \vspace{-0.35em}
\begin{align*}
    \Xi := \textrm{span} \{d \xi_1, \cdots, d \xi_r \}, \quad H := \textrm{span} \{d \eta_1, \cdots, d \eta_{n-r} \}. \vspace{-0.35em}
\end{align*}
The rank of $\Xi$ is $r$ \cite[Claim 9.5]{sastry2013nonlinear}.

% If $\eta$ satisfies Condition \ref{cond:nonsingular}, $H \cup \Xi = R^n$, and if it satisfies Condition \ref{cond:nonactuated_internal_state}, $H \subset G^\perp$. 

% \begin{lemma}
% \label{lemma:frobenius}
% (Frobenius Theorem \cite[Theorem 8.15]{sastry2013nonlinear}) Under Assumptions \ref{assmp:g_indepent} and \ref{assump:involutive}, we can find $\lambda:\R^n \rightarrow \R^{n-m}$ such that $\lambda_1, \cdots, \lambda_{n-m}$ are linearly independent, and $G^{\perp}=\textrm{span} \{\nabla \lambda_1, \cdots, \nabla \lambda_{n-m}\}$. That is, for all $i = 1, \cdots, m$, and $j =1, \cdots, n-m$,
% $$ L_{g_i} \lambda_j (x) = 0
% $$
% for all states in the open set around $x$.    
% \end{lemma}

% \begin{proof} This is a result of the .
% \end{proof}

\begin{lemma}
\label{lemma:equivalent_cond}
If $\eta: \R^n \rightarrow \R^{n-r}$ satisfies Conditions \ref{cond:nonsingular} and \ref{cond:nonactuated_internal_state}, \vspace{-0.35em}
\begin{equation}
G^{\perp} \cup \Xi = \R^n.
\label{eq:equivalent_cond} \vspace{-0.35em}
\end{equation}
\end{lemma}

\begin{proof} If $\eta$ satisfies Condition \ref{cond:nonsingular}, $H \cup \Xi = R^n$, and if it satisfies Condition \ref{cond:nonactuated_internal_state}, $H \subset G^\perp$. Thus, $G^{\perp} \cup \Xi = \R^n$. \vspace{-0.35em}
\end{proof}

\begin{lemma}
\label{lemma:existence_of_q}
If $g_1, \cdots, g_m$ are linearly independent at $x$, there \textit{always} exist a nontrivial $q \in \R^n$  such that $q \perp G^{\perp}$, and $q \perp \Xi$.
\end{lemma}

\begin{proof}
We consider a coefficient vector $c\in \R^m$ and set $q = [g_1, \cdots, g_m] c$. Such $q$ is not identically zero due to the linear independence of $g_i$'s. Since $q \in G$, $q \perp G^{\perp}$. For $q \perp \Xi$, we want to satisfy \vspace{-0.35em}
\begin{equation}
    \underbrace{\resizebox{.15\hsize}{!}{$\displaystyle\begin{bmatrix}dh \\ d L_f h \\ \vdots \\ d L_f^{r -1 } h \end{bmatrix}$}}_{=d \xi \in\R^{r \times n}} \underbrace{\underbrace{\begin{bmatrix} g_1 &  \cdots & g_m \end{bmatrix}}_{\in\R^{n \times m}} \begin{bmatrix} c_1 \\ \vdots \\ c_m \end{bmatrix}}_{=q} = 0. \label{eq:multiplication} \vspace{-0.35em}
\end{equation}
The multiplication of the first two matrices becomes
$$
\resizebox{.45\hsize}{!}{$\displaystyle\begin{bmatrix} 0 & \cdots & 0 \\ \; & \ddots & \; \\ 0 & \cdots & 0 \\ L_{g_1} L_f^{r -1 } h  & \cdots  & L_{g_m} L_f^{r -1 } h \end{bmatrix}$}
$$
where the top $r-1$ rows become zero due to the definition of the relative degree. Thus, by choosing $c$ orthogonal to $\begin{bmatrix} L_{g_1} L_f^{r -1 } h  & \cdots  & L_{g_m} L_f^{r -1 } h \end{bmatrix}$, \eqref{eq:multiplication} is satisfied.
\end{proof}

The proof of Theorem \ref{thm:multi-input-internal} results from Lemmas \ref{lemma:equivalent_cond} and \ref{lemma:existence_of_q}. If $\eta$ satisfies Conditions \ref{cond:nonsingular} and \ref{cond:nonactuated_internal_state}, by Lemma \ref{lemma:equivalent_cond}, $G^{\perp} \cup \Xi = \R^n$. This contradicts with Lemma \ref{lemma:existence_of_q} since Lemma \ref{lemma:existence_of_q} implies that $G^{\perp} \cup \Xi \neq \R^n$. \hfill $\qed$
% If $m > r$, the rank of $G^{\perp}$, $n-m$, is smaller than $n-r$, thus, since the rank of $Xi$ is $r$, $G^{\perp} \cup \Xi \subsetneq \R^n$. If $m \le r$, Lemma \ref{lemma:existence_of_q} implies that $G^{\perp} \cup \Xi \neq \R^n$. Thus, in both cases, by Lemma \ref{lemma:equivalent_cond}, $\eta$ cannot satisfy Conditions \ref{cond:nonsingular} and \ref{cond:nonactuated_internal_state} at the same time. 
% \end{proof}

\bibliographystyle{IEEEtran}
\bibliography{references}

\end{document}